\documentclass[journal]{IEEEtran}
\usepackage{flushend}
\usepackage{graphicx} 
\usepackage{subcaption}
\usepackage{tabularray}

\usepackage{tikz}
\usetikzlibrary{arrows.meta,decorations.pathreplacing,calc,positioning}

\usepackage[ruled,vlined,linesnumbered]{algorithm2e}
\usepackage{cite}
\usepackage{graphicx} 
\usepackage{subcaption}
\usepackage{cite}
\usepackage{hyperref}
\hypersetup{
    colorlinks,
    linkcolor={red!50!black},
    citecolor={blue!50!black},
    urlcolor={blue!50!black}
}
\usepackage{amsmath, amssymb, amsfonts, enumerate, bbold, multirow,booktabs,bm}

\newtheorem{assumption}{Assumption}
\newtheorem{theorem}{Theorem}
\newtheorem{lemma}{Lemma}
\newtheorem{proposition}{Proposition}
\newtheorem{corollary}{Corollary}

\newtheorem{remark}{Remark}
\newtheorem{definition}{Definition}

\newcommand{\rev}[1]{{\color{black}#1}}
\newcommand{\revv}[1]{{\color{black}#1}}

\usepackage[usenames,dvipsnames,table]{xcolor}
\usepackage{array} 

\begin{document}

\title{
Extended Version: Storage-Based Strategic Manipulation of Constraint-Binding Patterns in Power Networks
}

\author{Mehdi Davoudi, Minghao Mou, and Junjie Qin
\thanks{M. Davoudi, M. Mou, and J. Qin are with the Elmore School of Electrical and Computer Engineering, Purdue University, West Lafayette, IN, US. Emails:
        {\tt\small \{mdavoudi,mmou,jq\}@purdue.edu}}%
}


\maketitle

\addtolength{\textfloatsep}{-2mm}

\begin{abstract}
	This paper studies the strategic market participation of a monopolistic energy
	storage aggregator (ESA) in a day-ahead electricity market. The ESA coordinates
	geographically distributed storage units, submits a coordinated bid for its
	portfolio, and may hold financial transmission rights (FTRs). The system
	operator clears the market through a network-constrained, multi-period economic
	dispatch, determining generation and load schedules, nodal prices,
	energy-market payments, and FTR payoffs. We formulate the ESA--system-operator
	interaction as a Stackelberg game and characterize its equilibrium through a
	constraint-binding-pattern decomposition of the market-clearing problem.
	Beyond enabling equilibrium computation, the framework reveals how the ESA can
	increase its profit by strategically inducing or avoiding particular
	constraint-binding patterns. It also establishes a novel welfare result: although strategic storage without
	FTRs is known to weakly improve social welfare relative to the no-storage case,
	certain FTR positions can overturn this guarantee by strengthening the ESA's
	incentive to induce particular patterns, causing social welfare to fall below
	the no-storage level. Motivated
	by these findings, we develop two system-operator mechanisms for limiting
	undesirable ESA behavior and its adverse effects on market outcomes and social
	welfare. Finally, a three-bus study illustrates the theoretical findings, while IEEE test systems
	demonstrate the scalability of the proposed method.
\end{abstract}

\section{Introduction}\label{sec:intro}

Energy storage is playing an expanding role in power systems by providing
flexibility to integrate renewable generation, balance supply and demand, and
support reliability services \cite{denholm2010role}. Storage is deployed across different scales, from utility-scale batteries, such as those
in Texas \cite{eolian} and California \cite{mosslandings}, to smaller
behind-the-meter systems at residential and commercial sites. Flexible loads can also provide storage-like capabilities by shifting electricity consumption across time. One way to monetize the operational flexibility provided by these storage and storage-like resources is through participation in electricity markets, which creates revenue opportunities for resource owners and encourages further investment.

Although policies such as FERC Order No.~841 have facilitated such market
participation \cite{FERCOrder841}, many smaller resources remain unable to
participate effectively on their own. Aggregation therefore provides an important mechanism for enabling their participation, as further supported by FERC Order No.~2222 \cite{FERCOrder2222}. An energy storage aggregator (ESA) can coordinate collections of distributed storage and storage-like resources and participate in the market as a single entity with greater effective capacity than the individual resources alone.

Once aggregated storage resources participate in wholesale electricity markets, the ESA’s submitted charging and discharging schedules for these resources become inputs to the market-clearing process and can affect market outcomes. This creates an opportunity for the ESA to coordinate its storage resources strategically. When these resources are distributed across multiple buses, the ESA can coordinate its actions across both time and network locations, introducing a \emph{spatial} dimension to its strategic behavior beyond the \emph{temporal} dimension of conventional storage arbitrage.

This spatial dimension matters because electricity markets are cleared over a
network in which decisions at one location affect flows, congestion, and prices
elsewhere. A prominent manifestation of this coupling is transmission
congestion: when a line binds, demand that could otherwise be served by
lower-cost remote generation may require more expensive local generation,
raising locational marginal prices (LMPs) and congestion costs. Indeed, U.S. transmission-congestion costs
exceeded $\$11.5$ billion in 2023~\cite{gridstrategies2023}. More broadly, line congestion is one example of a binding constraint within a
market-clearing \emph{constraint-binding pattern}, defined by the set of
constraints active at the market-clearing optimum. Such patterns may also involve generators or elastic
loads reaching their operating limits, thereby affecting dispatch, LMPs, and
social welfare. The system operator may therefore regard some patterns as more desirable than
others or seek to avoid certain patterns based on operational or market-design
criteria.

A strategic ESA, however, may prefer different patterns and coordinate its
charging and discharging across time and buses to induce them. For example, causing a transmission line to become congested may create
favorable spatial or intertemporal price differences for charging and
discharging, even though the resulting congestion may be undesirable from the
system operator's perspective. We refer to such a divergence as
\emph{constraint-binding-pattern manipulation}
(see Definition~\ref{def:binding_pattern_manipulation} for the formal
definition). Characterizing this strategy becomes particularly challenging when the ESA
coordinates storage units spatiotemporally across multiple buses. This raises the central open question addressed in this paper:
\emph{How can an ESA use spatiotemporal storage coordination to influence
	market outcomes through constraint-binding-pattern manipulation, and how can
	system operators limit such behavior?}

\subsection{Organization and Contributions}
To answer this question, we study the strategic participation of a monopolistic
ESA that coordinates geographically distributed storage resources in a
day-ahead electricity market. 

Section~\ref{sec:model} formulates the interaction between the ESA and the
independent system operator (ISO), which clears the market, as a Stackelberg
game, where the ESA earns revenues from both energy arbitrage and financial transmission right (FTR) positions.
Section~\ref{parametric} exploits the
multiparametric structure of the market-clearing problem to characterize how
ESA schedules map to dispatch and LMPs across different constraint-binding
patterns. Building on this structure,
Section~\ref{sec:esa_binding_pattern_decomposition} decomposes the ESA's
strategic problem into convex subproblems associated with distinct
constraint-binding patterns and recovers the ESA--ISO equilibrium by comparing
their optimal solutions.
Section~\ref{sec:pattern_manipulation_ftr} formally defines
constraint-binding-pattern manipulation and characterizes how FTR positions can alter the ESA's
preferred binding patterns and the welfare implications of strategic storage
participation. Building on these findings,
Section~\ref{sec:iso_binding_pattern_control} develops two alternative ISO-side
tools for limiting undesirable binding patterns: restricting the submitted
charging and discharging schedules of an ESA with given bus-specific storage
capacities, or limiting the storage capacity that the ESA may aggregate at
each bus. Finally, Section~\ref{numerical} illustrates the theoretical findings through
numerical studies, and Section~\ref{sec:conclusion} concludes the paper.

We contribute to the literature in the following ways.

\begin{enumerate}

	\item To the best of our knowledge, this is the \emph{first} study to incorporate
	FTR payoffs into an ESA's strategic market-participation model. While prior
	work shows that strategically operated storage weakly improves welfare relative
	to the no-storage benchmark~\cite{contreras2017participation}, we show that this
	guarantee can fail when ESA holds FTR positions.
	
	\item Our constraint-binding-pattern decomposition explicitly characterizes
	how submitted ESA schedules map to dispatch quantities, LMP trajectories, and
	the ESA's profit across distinct binding regimes. In contrast, this relationship remains implicit in mathematical program with
	equilibrium constraints (MPEC) reformulations, commonly used to compute
	ESA--ISO Stackelberg equilibria, e.g.,
	\cite{mohsenian2015coordinated,hartwig2015impact,huang2017market,
		bjorndal2023energy}.
	This explicit schedule-to-market-outcome characterization explains why the ESA
	may strategically induce particular binding patterns, identifies its preferred
	alternatives across patterns, and allows the ISO to assess how restrictions may
	alter the ESA's strategies.
	
	\item Building on this framework, we develop ex ante ISO-side tools that limit
	undesirable constraint-binding patterns while preserving the ESA's autonomy to
	choose among admissible schedules. Because these
	tools operate through participation constraints rather than changes to
	settlement prices, they can be incorporated into existing market frameworks
	more readily than approaches based on specialized pricing rules,
	e.g.,~\cite{contreras2017participation}.
\end{enumerate}

This paper extends our conference work~\cite{davoudi2026strategic} by
(a) generalizing line-congestion to constraint-binding-pattern manipulation,
(b) replacing the MPEC with an explicit binding-pattern decomposition of the
Stackelberg equilibrium, and (c) developing ISO-side tools to limit undesirable
ESA strategies.

\subsection{Related Literature}

A predominant stream of literature studies the strategic behavior of energy storage resources from a temporal-arbitrage perspective at a single location. For example, Williams and Green \cite{williams2022electricity} compare storage and generator market power, Karaduman \cite{karaduman2021economics} analyzes price-making storage in a dynamic equilibrium framework, and Wu et al.~\cite{wu2024market} examine strategic capacity withholding and its effects on
social welfare.

In contrast, relatively few studies investigate the strategic behavior of ESAs
that coordinate geographically distributed storage resources in electricity markets.
In this regard, Mohsenian-Rad \cite{mohsenian2015coordinated} develops a
framework for a price-making ESA coordinating geographically distributed
storage and shows that such coordination can significantly affect nodal market
outcomes. Hartwig and Kockar \cite{hartwig2015impact} study strategic storage operation
under different ownership arrangements, including portfolios combining storage
and conventional generation, and show that profit-maximizing storage owners may
withhold capacity, reducing social welfare and limiting the benefits storage can
provide to the system. Contreras-Ocaña et al. \cite{contreras2017participation} examine ESA
participation in wholesale electricity markets and propose a nonuniform pricing
method to mitigate the ESA's market power and better align its decisions with
social welfare.
Huang et al. \cite{huang2017market} compare centralized, semi-centralized, and deregulated market mechanisms for a price-maker ESA, which differ in the degree of operational control assigned to the ESA versus the system operator. Extending the analysis to strategic competition among multiple ESAs, Huang et al. \cite{9406390} model storage participation as a networked Cournot game and show that while strategic storage operation can reduce welfare relative to the social optimum, welfare losses diminish as market competition increases. Bj{\o}rndal et al. \cite{bjorndal2023energy} develop a multistage Stackelberg
framework for a monopolistic storage aggregator participating in day-ahead and
real-time markets and examine how nodal, zonal, and single-zone pricing affect
storage incentives and market outcomes.

Despite these contributions, existing studies largely evaluate strategic ESA
behavior through equilibrium prices, profits, and social welfare. While these
outcomes quantify the consequences of strategic behavior, they provide limited
guidance for ex ante mitigation of its potential adverse impacts. Existing
operator-side approaches are also limited, with proposed methods such as
nonuniform pricing not readily aligned with prevailing wholesale-market pricing
structures. This limitation partly arises because prior studies have primarily
considered capacity withholding as the key strategic mechanism and have not explicitly
examined constraint-binding-pattern manipulation as a distinct strategic
mechanism. Characterizing such manipulation links ESA actions to the binding
patterns they can induce, providing a basis for targeted ex ante restrictions.
Moreover, the role of FTR positions in strategic ESA
behavior has not been examined. Because FTR payoffs depend on congestion and
nodal price differences, they may create additional incentives to induce
particular binding patterns \cite{rosellon2013financial}. Thus, the interaction
between FTR positions and constraint-binding-pattern manipulation also remains
unexplored.

\section{Problem Formulation}\label{sec:model}

\subsection{Problem Setup}

We consider a day-ahead wholesale electricity market over a finite discrete-time
horizon $\mathcal{T}:=\{1,2,\ldots,T\}$, indexed by $t$, with each period
representing one hour, and operated on a transmission network represented by a
set of buses $\mathcal{N}$, indexed by $i$, and a set of directed transmission
lines $\mathcal{L}$. Each transmission line connecting buses $i$ and $j$ is denoted by
$(i,j)\in\mathcal L$, with $i<j$, where the direction from $i$ to $j$ is chosen as the positive reference direction for line flow.
Each bus may host generators, loads, and storage resources. Generators and loads are modeled as non-strategic participants that submit their cost and utility parameters truthfully, respectively.

A strategic ESA controls storage resources distributed across multiple network locations. To isolate the fundamental incentives of strategic storage operation, we consider a single monopolistic ESA and leave competition among multiple ESAs for future work. The storage resources controlled by the ESA at each bus are represented by an equivalent aggregate storage resource with aggregate capacity and operating limits.

The ESA participates in the market through \emph{self-schedule bids}, rather than
\emph{economic bids}, which specify both quantities and explicit price components
\cite{mohsenian2015coordinated}. Specifically, prior to market clearing, the ESA
submits a quantity-only bid $u_i(t)$ for each bus $i\in\mathcal{N}$ and time period
$t\in\mathcal{T}$, where $u_i(t)>0$ corresponds to charging and $u_i(t)<0$
corresponds to discharging.

\begin{remark}[Self-schedule bids]
The self-schedule formulation treats the ESA's submitted charging and discharging quantities as effectively inelastic inputs to the ISO's market-clearing problem. Equivalently, the self-scheduled quantities can be interpreted as arising from
economic bids submitted at sufficiently favorable prices, so that the
corresponding schedule is cleared whenever it is feasible for market clearing.
\end{remark}

\textbf{Notation:} For a bus- and time-indexed variable $x_i(t)$, we use $\mathbf{x}$ to denote the collection of all values $\{x_i(t)\}_{i\in\mathcal{N},\,t\in\mathcal{T}}$. The vector $\mathbf{x}_i$ collects the values associated with bus $i$ over all time periods, while $\mathbf{x}(t)$ collects the values across all buses at time $t$. For a time-indexed variable $y(t)$, we use $\mathbf{y}$ to denote the vector of values over $\mathcal{T}$. We use $\mathbf{0}$ and $\mathbf{1}$ to denote all-zero and all-one vectors of appropriate dimensions, respectively, $\mathrm{diag}(\cdot)$ to denote the diagonal matrix formed from its argument, and $|\mathcal S|$ to denote the cardinality of a set $\mathcal S$.

\subsection{ISO's Market Clearing Problem}

Given the submitted ESA bids, together with the cost and utility parameters
provided by generators and loads, the ISO clears the market by solving a
multi-period economic dispatch problem subject to network constraints. Before
presenting the complete market-clearing model, we introduce its three main
components: generation costs, demand utilities, and transmission-network
constraints.

\subsubsection{Generators}

Generator production costs are modeled using quadratic functions. Let $\bm{\alpha}^{\mathrm{g}}$ and $\bm{\beta}^{\mathrm{g}}$ denote the matrices collecting the quadratic and linear cost coefficients, respectively. For each time period $t\in\mathcal{T}$, the aggregate generation cost is given by
\begin{equation}
	C_t\left(\mathbf{g}(t)\right)
	:= \frac{1}{2}\mathbf{g}(t)^\top
	\mathrm{diag}\!\left(\bm{\alpha}^{\mathrm{g}}(t)\right)\mathbf{g}(t)
	+ \bm{\beta}^{\mathrm{g}}(t)^\top \mathbf{g}(t),\nonumber
\end{equation}
where $\mathbf g(t)\in\mathbb R_{+}^{|\mathcal N|}
$ denotes the vector of generation outputs. Generator dispatch is constrained by
\begin{equation}\label{DCOPF:generator}
	\underline{\mathbf{g}}(t) \leq \mathbf{g}(t) \leq \overline{\mathbf{g}}(t),
	\quad \forall t \in \mathcal{T},
\end{equation}
where
$
\underline{\mathbf g}(t),
\overline{\mathbf g}(t)
\in
\mathbb R_{+}^{|\mathcal N|}
$
denote the minimum and maximum generation limits, respectively. Renewables are fixed at their available output by setting both generation
bounds equal to that output; both bounds are zero at nongenerator buses.

\subsubsection{Loads}

Elastic demand is represented through quadratic utility functions. Let $\bm{\alpha}^{\mathrm{d}}$ and $\bm{\beta}^{\mathrm{d}}$ collect the corresponding utility coefficients. The aggregate utility obtained from served demand at time $t$ is
\begin{equation}
	V_t\left(\mathbf{d}(t)\right)
	:= -\frac{1}{2}\,\mathbf{d}(t)^\top
	\mathrm{diag}\!\left(\bm{\alpha}^{\mathrm{d}}(t)\right)\mathbf{d}(t)
	+ \bm{\beta}^{\mathrm{d}}(t)^\top \mathbf{d}(t),\nonumber
\end{equation}
where $ \mathbf d(t)\in\mathbb R_{+}^{|\mathcal N|} $ denotes the vector of served demand quantities. Demand is restricted by
\begin{equation}\label{DCOPF:elastic_load}
	\underline{\mathbf{d}}(t)
	\leq
	\mathbf{d}(t)
	\leq
	\overline{\mathbf{d}}(t),
	\quad \forall t \in \mathcal{T},
\end{equation}
where $ \underline{\mathbf d}(t), \overline{\mathbf d}(t) \in \mathbb R_{+}^{|\mathcal N|} $ denote the lower and upper consumption limits, respectively. Inelastic demand is modeled by setting these bounds equal to the prescribed demand level, whereas buses without demand have both limits set to zero.

\subsubsection{Network constraints}

The net injection at each bus is
\begin{equation}
	\mathbf{p}(t)
	=
	\mathbf{g}(t)
	-
	\mathbf{d}(t)
	-
	\mathbf{u}(t),
	\quad \forall t \in \mathcal{T},
	\label{DCOPF:net_power}
\end{equation}
where positive values correspond to net supply and negative values correspond to net consumption. Since transmission losses are neglected, the system-wide power balance requires
\begin{equation}
	\mathbf{1}^\top \mathbf{p}(t)
	=
	0,
	\quad \forall t \in \mathcal{T}.
	\label{DCOPF:lossless}
\end{equation}

Line flows are represented through the Power Transfer Distribution Factor (PTDF) matrix $\mathbf{H}$ and constrained by
\begin{equation}
	-\overline{\mathbf{f}}
	\leq
	\mathbf{H}\mathbf{p}(t)
	\leq
	\overline{\mathbf{f}},
	\quad \forall t \in \mathcal{T},
	\label{DCOPF:line}
\end{equation}
where $\overline{\mathbf f}\in\mathbb R_{+}^{|\mathcal L|}$ denotes the vector
of thermal line limits, and the corresponding PTDF matrix is denoted by
$\mathbf H\in\mathbb R^{|\mathcal L|\times|\mathcal N|}$.

Given the submitted bids, the ISO clears the market by minimizing generation
costs minus consumer utility, subject to the generator, demand, and network
constraints. This is equivalent to maximizing social welfare. Combining the
preceding components yields the following market-clearing problem:

\begin{subequations}\label{DCOPF}
	\begin{alignat}{2}
		\min_{\mathbf{g},\,\mathbf{d},\,\mathbf{p}}
		\quad
		& & &
		\sum_{t \in \mathcal{T}}
		\left[
		C_t\left(\mathbf{g}(t)\right)
		-
		V_t\left(\mathbf{d}(t)\right)
		\right]
		\label{DCOPF:obj}
		\\
		\text{s.t.}
		\quad \,\,
		&\bm{\lambda}(t)\!\!:\,
		& &
		\mathbf{p}(t)
		=
		\mathbf{g}(t)
		-
		\mathbf{d}(t)
		-
		\mathbf{u}(t),
		\quad
		\forall t \in \mathcal{T},
		\label{DCOPF:net_power:full}
		\\
		& & &
\eqref{DCOPF:generator},
~
\eqref{DCOPF:elastic_load},
~
\eqref{DCOPF:lossless},
~
\eqref{DCOPF:line},
		\label{DCOPF:Constraints}
	\end{alignat}
\end{subequations}
where $\bm{\lambda}(t)$ denotes the Lagrange multiplier associated with the
nodal power-balance constraint~\eqref{DCOPF:net_power:full}. These multipliers
are taken as the cleared market prices, namely LMPs, at time $t$.

\subsection{ESA's Market Participation Problem}
\label{market:participation}

The ESA's self-schedule bids affect dispatch and LMPs through the ISO
market-clearing problem~\eqref{DCOPF}, while the resulting LMPs determine its
net energy-market payment and FTR payoff. \revv{We model this strategic interdependence as a Stackelberg game, where the ESA is the leader because it submits its schedule before market clearing, while the ISO is the follower because it clears the market in response to that schedule.}
To formulate the game, we first
introduce the ESA's objective components and operational constraints and then
present the resulting ESA--ISO Stackelberg problem.

\subsubsection{ESA objective components}

The ESA's objective consists of two components: its net energy-market payment
and its FTR payoff.

\paragraph{Net energy-market payment}
Under the sign convention $u_i(t)>0$ for charging and $u_i(t)<0$ for
discharging, the ESA's net energy-market payment is
$\sum_{t\in\mathcal T}\bm{\lambda}(t)^\top\mathbf u(t)$.

\paragraph{FTR payoff}
Financial transmission rights (FTRs) are financial instruments that allow market participants to hedge against congestion-related price differences in the transmission network \cite{rosellon2013financial}. FTRs can be formulated as point-to-point rights or flowgate rights. Point-to-point FTRs are defined between a source bus and a sink bus, whereas flowgate rights are associated with congestion on a specific transmission line \cite{rosellon2013financial,sarkar2008comprehensive}.
Point-to-point FTRs are commonly structured as either obligations or options. An FTR obligation entitles its holder to the nodal price difference between the sink and source buses and may therefore result in either a positive or negative payoff. In contrast, an FTR option pays only when the sink price exceeds the source price and therefore cannot create a financial liability for its holder \cite{li2005risk}. In this paper, we consider point-to-point FTR obligations.

Let $\mathcal{F}\subseteq\mathcal{N}\times\mathcal{N}$ denote the set of source-sink FTR pairs, where $(i,j)\in\mathcal{F}$ represents an FTR from source bus $i$ to sink bus $j$. Furthermore, let $\kappa_{(i,j)}(t)$ denote the quantity of FTR obligations held on path $(i,j)$ at time $t$, and $
\bm{\kappa}_{(i,j)}
:=
\big[\kappa_{(i,j)}(t)\big]_{t\in\mathcal T}.
$
The corresponding aggregate FTR payoff is then given by
$
\sum_{(i,j)\in\mathcal F}
\bm{\kappa}_{(i,j)}^\top
\left(
\bm{\lambda}_j-\bm{\lambda}_i
\right).
$

In practice, FTR positions are obtained through ISO-administered auctions subject to simultaneous-feasibility and revenue-adequacy requirements \cite{sarkar2008comprehensive}. We abstract from the auction-clearing process and treat $\bm{\kappa}$ as an exogenous parameter. The no-FTR case is recovered by setting $\bm{\kappa}=\mathbf 0$.

\subsubsection{Storage constraints}

Storage is assumed ideal, with unit charging and discharging efficiencies and
zero initial stored energy. Let $B_i$ denote the aggregated energy capacity at
bus $i$, with $B_i=0$ at buses without storage. The corresponding
state-of-charge trajectory satisfies
\begin{equation}
	\mathbf{0}
	\le
	\mathbf{L}\mathbf{u}_i
	\le
	B_i\mathbf{1},
	\qquad
	\forall i\in\mathcal N,
	\label{ESA:bilevel:dynamic}
\end{equation}
where $\mathbf{L}\in\mathbb{R}^{|\mathcal T|\times|\mathcal T|}$ is lower
triangular with ones on and below the diagonal.
To enforce cyclic operation over the study horizon, we impose
\begin{equation}
	\mathbf{1}^{\top}\mathbf{u}_i=0,
	\qquad
	\forall i\in\mathcal N,
	\label{ESA:bilevel:cyclic}
\end{equation}
which requires each storage resource to finish the horizon at its initial state of charge. 
Together, constraints~\eqref{ESA:bilevel:dynamic} and~\eqref{ESA:bilevel:cyclic} define the feasible-action set of the ESA
\begin{equation}
\!\!\mathcal U(\mathbf B)
	:=
	\left\{
	\mathbf u:\!
	\mathbf 0
	\leq
	\mathbf L\mathbf u_i
	\leq
	B_i\mathbf 1,
	\quad \!
	\mathbf 1^\top\mathbf u_i
	=
	0,
	\quad \! \!
	\forall i\in\mathcal N
	\right\}\!.\nonumber
\end{equation}

\subsubsection{ESA--ISO Stackelberg formulation}
Combining the ESA's objective components and storage constraints yields the
following ESA--ISO Stackelberg problem:
\begin{subequations}
	\label{bilevel}
	\begin{align}
		\min_{\mathbf u}
		\quad
		&
		\sum_{t\in\mathcal T}
		\bm{\lambda}(t)^\top\mathbf u(t)
		-
		\sum_{(i,j)\in\mathcal F}
		\bm{\kappa}_{(i,j)}^\top
		\left(
		\bm{\lambda}_j
		-
		\bm{\lambda}_i
		\right)
		\label{bilevel:obj}
		\\
		\text{s.t.}
		\quad
		&
		\mathbf u
		\in
		\mathcal U(\mathbf B),
		\label{bilevel:storage}
		\\
		&
		\bm{\lambda}
		\text{ is an optimal Lagrange multiplier for~\eqref{DCOPF} given }
		\mathbf u.
		\label{bilevel:market_response}
	\end{align}
\end{subequations}

Problem~\eqref{bilevel} represents the ESA--ISO Stackelberg interaction: the ESA
selects its charging and discharging schedule while accounting for the ISO's
response and resulting prices. Constraint~\eqref{bilevel:market_response}
embeds this response as the lower-level problem, making~\eqref{bilevel} a
bilevel optimization problem. A conventional solution approach replaces the
lower-level problem with its optimality conditions, yielding a single-level
MPEC~\cite{davoudi2026strategic}. However, this reformulation leaves implicit how
the ESA's submitted schedule impacts the resulting market outcomes.
We instead explicitly characterize this schedule-to-market-response mapping by
treating the ESA's submitted schedule as a parameter of~\eqref{DCOPF} and
exploiting its multiparametric structure. This partitions the ESA schedule
space according to the ISO's constraint-binding patterns, characterizes market responses within each region, and decomposes
\eqref{bilevel} into pattern-specific subproblems. This characterization shows how binding constraints shape the ESA's objective
and pattern preferences. The next section develops this characterization.

\section{Market-Clearing Responses Across Constraint-Binding Patterns}
\label{parametric}

Let $\mathcal U_{\mathrm{mc}}\subseteq\mathbb R^{|\mathcal N|\times|\mathcal T|}$ denote
the feasible parameter set of the market-clearing problem, i.e., the set of
ESA schedules $\mathbf u$ for which~\eqref{DCOPF} is feasible. This set is
independent of any particular storage-capacity vector $\mathbf B$ and
characterizes the full range of schedules that can be cleared by the ISO. An
ESA with a given capacity vector $\mathbf B$, however, can submit only
schedules in
$
\mathcal U(\mathbf B)\cap\mathcal U_{\mathrm{mc}}. 
$
\revv{To characterize the ISO's market-clearing response over $\mathcal U_{\mathrm{mc}}$, we
adopt the following assumptions throughout the remainder of the paper.}

\begin{assumption}[Strict convexity]
	\label{assum1}
	For all $t\in\mathcal T$, the quadratic coefficients satisfy
	$
	\bm{\alpha}^{\mathrm g}(t)>\mathbf 0
	$
	and
	$
	\bm{\alpha}^{\mathrm d}(t)>\mathbf 0,
	$
	componentwise.
\end{assumption}

\begin{assumption}[LICQ]
	\label{assum2}
	For every $\mathbf u\in\mathcal U_{\mathrm{mc}}$, the active constraints at the
	corresponding market-clearing optimum satisfy the Linear Independence
	Constraint Qualification (LICQ).\footnote{Although we use equal upper and lower bounds in
		~\eqref{DCOPF:generator} and~\eqref{DCOPF:elastic_load} to fix renewable
		generation and inelastic demand, these constraints are excluded from the LICQ
		test because the corresponding quantities are fixed inputs to~\eqref{DCOPF}
		and are modeled as variables only for notational consistency.}
\end{assumption}

\begin{assumption}[Market-clearing feasibility]
	\label{assum3}
	The submitted ESA schedule satisfies
	$
	\mathbf u
	\in
	\mathcal U(\mathbf B)
	\cap
	\mathcal U_{\mathrm{mc}}.
	$
\end{assumption}

Assumption~\ref{assum1} ensures uniqueness of the market-clearing primal
solution, while Assumption~\ref{assum2} ensures uniqueness of the corresponding
Lagrange multipliers and a regular Karush--Kuhn--Tucker (KKT) representation.
Assumption~\ref{assum3} reflects rational ESA behavior: the ESA has no incentive
to submit a schedule that the ISO cannot feasibly clear.

With the market-clearing response uniquely defined under
Assumptions~\ref{assum1}--\ref{assum2}, for each
$\mathbf u\in\mathcal U_{\mathrm{mc}}$ and $t\in\mathcal T$, define
\begin{equation}
	\mathbf x^\star(t;\mathbf u)
	:=
	\left[
	\mathbf g^\star(t;\mathbf u)^\top,
	\mathbf d^\star(t;\mathbf u)^\top,
	\mathbf p^\star(t;\mathbf u)^\top
	\right]^\top,\nonumber
\end{equation}
and let
$
\bm{\lambda}^\star(t;\mathbf u)
$
denote the corresponding LMP vector. The market-clearing response is
characterized by the mapping
$
\mathbf u
\mapsto
\left\{
\mathbf x^\star(t;\mathbf u),
\bm{\lambda}^\star(t;\mathbf u)
\right\}_{t\in\mathcal T}.
$
Since the ESA's objective depends on the resulting LMPs, the price-response
component of this mapping is of particular interest. To characterize this mapping, we first formalize the constraint-binding patterns
induced by ESA schedules and then derive the associated responses.

\subsection{Critical Regions and Constraint-Binding Patterns}
\label{subsec:critical_regions}

Different submitted ESA schedules $\mathbf u$ may induce different sets of
binding constraints at the market-clearing optimum. We formalize this
dependence through the following definitions.

\begin{definition}[Market-clearing constraint-binding pattern]
	\label{def:market_clearing_binding_pattern}
	For a given ESA schedule $\mathbf u\in\mathcal U_{\mathrm{mc}}$, the
	\emph{market-clearing constraint-binding pattern}, denoted by
	$\mathcal A(\mathbf u)$, is the index set of inequality constraints
	in~\eqref{DCOPF} that are active at the market-clearing optimum.
\end{definition}

In particular, each pattern identifies which transmission lines are congested
and which generators and elastic loads operate at their lower or upper limits.

\begin{definition}[Critical region]
	\label{def:critical_region}
A \emph{critical region} $\mathcal C_r\subseteq\mathcal U_{\mathrm{mc}}$ is the closure
of a maximal connected subset of $\mathcal U_{\mathrm{mc}}$ over whose relative interior
$\mathcal A(\mathbf u)$ remains constant.
\end{definition}

Let $R$ denote the total number of critical regions and
$\mathcal R:=\{1,\ldots,R\}$ their index set. These regions collectively cover
the parameter space,
$
\mathcal U_{\mathrm{mc}}=\bigcup_{r\in\mathcal R}\mathcal C_r,
$
and adjacent regions may share boundaries.

\revv{To formulate the ESA problem separately over each critical region, we need an
	explicit description of the schedules associated with that region. The
	following result shows that each critical region can be represented by linear
	inequalities in the submitted ESA schedule. For compact notation, define
\[
\operatorname{vec}(\mathbf u)
:=
\bigl[
\mathbf u_1^\top
~
\cdots
~
\mathbf u_{|\mathcal N|}^\top
\bigr]^\top
\in
\mathbb R^{|\mathcal N||\mathcal T|}.
\]}

\revv{\begin{lemma}[Polyhedral representation of a critical region]
		\label{lem:critical_region_polyhedron}
		For each $r\in\mathcal R$, there exist a matrix $\mathbf F_r$ and a vector
		$\mathbf h_r$ such that
		\begin{equation}
			\mathcal C_r
			=
			\left\{
			\mathbf u:
			\mathbf F_r
			\operatorname{vec}(\mathbf u)
			\leq
			\mathbf h_r
			\right\}.
			\label{eq:critical_region_polyhedral}
		\end{equation}
	\end{lemma}
}

\revv{The representation in~\eqref{eq:critical_region_polyhedral} follows by fixing
the active set associated with region $r$, solving the corresponding KKT
system, and imposing primal and dual feasibility. \emph{
	Detailed proofs of the analytical results are provided in the appendices.} In the following subsection, we
characterize the dispatch and LMP responses within and across critical regions.}

\subsection{Market-Clearing Response Characterization}
\label{subsec:affine_market_clearing_response}

\revv{Having partitioned the feasible parameter space according to
	constraint-binding patterns, we now characterize the ISO's response within each
	critical region, where the pattern is fixed. The following theorem shows how
	dispatch and LMPs depend on the submitted ESA schedule and how these responses
	connect across region boundaries.}

\begin{theorem}[Continuous piecewise-affine market-clearing response]
	\label{thm:affine_mapping}
	Under Assumptions~\ref{assum1}--\ref{assum3}, the following hold:
	\begin{enumerate}
		\item For each critical region $\mathcal C_r$, $r\in\mathcal R$,
		and time period $t\in\mathcal T$, there exist matrices
		$\bm{\Gamma}_r(t)$ and $\mathbf K_r(t)$, and vectors
		$\bm{\gamma}_r(t)$ and $\mathbf k_r(t)$, such that
		\begin{align}
			\mathbf x^\star(t;\mathbf u)
			&=
			\bm{\Gamma}_r(t)\mathbf u(t)
			+
			\bm{\gamma}_r(t),
			\qquad
			\forall \mathbf u\in\mathcal C_r,
			\label{eq:primal_affine}
			\\
			\bm{\lambda}^\star(t;\mathbf u)
			&=
			\mathbf K_r(t)\mathbf u(t)
			+
			\mathbf k_r(t),
			\qquad
			\forall \mathbf u\in\mathcal C_r.
			\label{eq:lambda_affine}
		\end{align}
		
		\item For every $t\in\mathcal T$, the market-clearing response maps
		\[
		\mathbf u
		\mapsto
		\mathbf x^\star(t;\mathbf u),
		\qquad
		\mathbf u
		\mapsto
		\bm{\lambda}^\star(t;\mathbf u),
		\]
		are continuous and piecewise affine on $\mathcal U_{\mathrm{mc}}$. In particular,
		the regional affine representations agree on the boundaries of
		adjacent critical regions.
		
		\item For every $r\in\mathcal R$ and $t\in\mathcal T$, the matrix
		$\mathbf K_r(t)$ is symmetric positive semidefinite.
	\end{enumerate}
\end{theorem}

The proof follows by solving the KKT system of~\eqref{DCOPF} under each fixed
active set. Theorem~\ref{thm:affine_mapping} transforms the
ISO's market-clearing response from an implicit lower-level outcome into an
explicit region-by-region mapping. Within each critical region, the associated
constraint-binding pattern determines how the submitted ESA schedule affects
dispatch quantities, LMPs, and, ultimately, the ESA's profitability. Moving to
another critical region changes this mapping by activating a different set of
constraints, thereby explaining why the ESA may strategically induce one
binding pattern over another. Continuity prevents jumps at shared
boundaries, while the positive semidefiniteness of $\mathbf K_r(t)$ yields a
convex ESA objective within each region. These properties support the critical-region decomposition of~\eqref{bilevel}
and the subsequent operator-side tools.

\section{Constraint-Binding-Pattern Decomposition of the ESA Strategic Market-Participation Problem}
\label{sec:esa_binding_pattern_decomposition}

This section develops a constraint-binding-pattern decomposition of the
ESA strategic market-participation problem and uses it to recover an
equilibrium of the ESA--ISO Stackelberg interaction. We proceed in three
steps. First, building on the piecewise-affine market-clearing response
established in Theorem~\ref{thm:affine_mapping}, we characterize the ESA's
objective within each critical region. Second, we formulate the corresponding
pattern-specific subproblems. Finally, we recover the ESA's globally optimal
strategy by comparing their optimal values.

\subsection{Critical-Region Characterization of the ESA Objective}
\label{subsec:regional_esa_objective}

To replace the implicit price-response condition in~\eqref{bilevel:market_response} with an explicit regional representation, we use the mapping established in Theorem~\ref{thm:affine_mapping}. To this end, we first write the ESA objective as a function of the submitted schedule:
\begin{equation}
	\!J(\mathbf u)
	\!:=\!\!
	\sum_{t\in\mathcal T}
	\!\left[
	\!\bm{\lambda}^{\star}\!\left(t;\mathbf u\right)^{\!\top}
	\!\!\mathbf u(t)
	\!-\!\!\!\!\!
	\sum_{(i,j)\in\mathcal F}
	\!\!\!\!\!\kappa_{(i,j)}(t)
	\!\left(
	\lambda_j^{\star}\!\left(t;\mathbf u\right)
	-
	\lambda_i^{\star}\!\left(t;\mathbf u\right)
	\right)
\!\right]\!\!.
	\label{eq:esa_objective_global}
\end{equation}

\revv{To formulate the pattern-specific ESA subproblems, we need an explicit
	expression for $J(\mathbf u)$ within each critical region. The following lemma
	provides this characterization. Let
	$\mathbf e_i\in\mathbb R^{|\mathcal N|}$ denote the vector with a one in its
	$i$-th entry and zeros elsewhere.}

\begin{lemma}[Convex quadratic ESA objective within each critical region]
	\label{lem:regional_esa_objective}
	For each critical region $\mathcal C_r$, $r\in\mathcal R$, the ESA
	objective admits the convex quadratic representation
	\begin{equation}
		J_r(\mathbf u)
		\!:=\!
		\sum_{t\in\mathcal T}
		\!\left[
		\mathbf u(t)^{\!\top}
		\mathbf K_r(t)
		\mathbf u(t)
		\!+
		\bm{\psi}_r(t)^{\!\top}
		\mathbf u(t)
		+
		\zeta_r(t)
		\right]\!,
		\forall \mathbf u\in\mathcal C_r,
		\label{eq:regional_esa_objective_quadratic}
	\end{equation}
	where
	\revv{\begin{align}
		\bm{\psi}_r(t)
		&:=
		\mathbf k_r(t)
		-
		\mathbf K_r(t)^{\!\top}
		\sum_{(i,j)\in\mathcal F}
		\kappa_{(i,j)}(t)
		\left(
		\mathbf e_j-\mathbf e_i
		\right),
		\label{eq:regional_linear_term}
		\\
		\zeta_r(t)
		&:=
		-
		\sum_{(i,j)\in\mathcal F}
		\kappa_{(i,j)}(t)
		\left(
		\mathbf e_j-\mathbf e_i
		\right)^{\!\top}
		\mathbf k_r(t).
		\label{eq:regional_constant_term}
	\end{align}}
\end{lemma}

\revv{Lemma~\ref{lem:regional_esa_objective} follows by substituting
	\eqref{eq:lambda_affine} into~\eqref{eq:esa_objective_global} and collecting
	the quadratic, linear, and constant terms in $\mathbf u$. The positive
	semidefiniteness of $\mathbf K_r(t)$ makes the objective within each critical
	region convex, which is key to the tractability of the pattern-specific ESA
	subproblems. Because adjacent critical regions may share boundaries, we must also verify
	that their quadratic representations agree wherever they overlap, as
	established by the following corollary.}

\begin{corollary}[Continuous piecewise-quadratic ESA objective]
	\label{cor:esa_objective_continuity}
	Under Assumptions~\ref{assum1}--\ref{assum3}, the ESA objective $J(\mathbf u)$ is a continuous
	piecewise-quadratic function on $\mathcal U_{\mathrm{mc}}$. More precisely,
	\begin{equation}
		J(\mathbf u)
		=
		J_r(\mathbf u),
		\qquad
		\forall \mathbf u\in\mathcal C_r,
		\quad
		r\in\mathcal R.\nonumber
	\end{equation}
	Furthermore, for any two adjacent critical regions
	$\mathcal C_r$ and $\mathcal C_s$,
	\begin{equation}
		J_r(\mathbf u)
		=
		J_s(\mathbf u),
		\qquad
		\forall \mathbf u
		\in
		\mathcal C_r
		\cap
		\mathcal C_s.\nonumber
	\end{equation}
\end{corollary}

\revv{Corollary~\ref{cor:esa_objective_continuity} follows directly from the
continuity of the price-response mapping established in
Theorem~\ref{thm:affine_mapping}. Thus, the functions
$\{J_r\}_{r\in\mathcal R}$ are not independent objectives; rather, they form
the quadratic pieces of a single continuous ESA objective. This structure
allows the ESA problem to be formulated separately within each critical
region, as done next.
}

\subsection{ESA Subproblems within Critical Regions}
\label{subsec:regional_esa_problems}

Building on Lemma~\ref{lem:regional_esa_objective} and
Corollary~\ref{cor:esa_objective_continuity}, we decompose the ESA problem into
one subproblem per critical region. For each $r\in\mathcal R$, restricting the
submitted ESA schedule to $\mathcal C_r$ yields:
\begin{subequations}
	\label{eq:regional_esa_qp}
	\begin{align}
		\min_{\mathbf u}\quad
		& J_r(\mathbf u)
		\label{eq:regional_esa_qp_obj}
		\\
		\text{s.t.}\quad
		& \mathbf u
		\in
		\mathcal U(\mathbf B),
		\label{eq:regional_esa_qp_storage}
		\\
		& \mathbf F_r
		\operatorname{vec}(\mathbf u)
		\leq
		\mathbf h_r.
		\label{eq:regional_esa_qp_region}
	\end{align}
\end{subequations}

\revv{Constraint~\eqref{eq:regional_esa_qp_region} ensures that the submitted
	ESA schedule lies in the closed critical region $\mathcal C_r$. If
	$\mathcal U(\mathbf B)\cap\mathcal C_r=\varnothing$, then
	problem~\eqref{eq:regional_esa_qp} is infeasible, and its optimal value is
	taken to be $+\infty$. By Lemma~\ref{lem:regional_esa_objective} and the linear constraints in
	\eqref{eq:regional_esa_qp_storage}--\eqref{eq:regional_esa_qp_region},
	problem~\eqref{eq:regional_esa_qp} is a convex quadratic program. Its objective,
	however, need not be strictly convex, so a feasible problem may admit multiple
	optimal schedules. This raises the concern that the corresponding LMP trajectory
	may not be uniquely determined within each critical region. The following lemma
	shows that, despite this nonuniqueness in the optimal schedule, the LMP trajectory remains unique.}

\revv{\begin{lemma}[LMP uniqueness within a critical-region subproblem]
	\label{lem:regional_price_uniqueness}
	For a fixed storage-capacity vector $\mathbf B$ and critical region
	$r\in\mathcal R$, define the optimizer set as
\begin{equation}
	\mathcal U_r^\star(\mathbf B)
	:=
	\operatorname*{arg\,min}_{\substack{
			\mathbf u\in\mathcal U(\mathbf B)\\
			\mathbf F_r\operatorname{vec}(\mathbf u)\leq\mathbf h_r
	}}
	J_r(\mathbf u).\nonumber
\end{equation}
If $\mathcal U_r^\star(\mathbf B)$ is nonempty, then all
$\mathbf u\in\mathcal U_r^\star(\mathbf B)$ induce the same cleared LMP
trajectory.
\end{lemma}
}

\revv{Lemma~\ref{lem:regional_price_uniqueness} establishes that each feasible
	critical-region subproblem has a well-defined optimal LMP trajectory, even when
	its optimal ESA schedule is not unique. Thus, the outcome of each subproblem can
	be characterized unambiguously in terms of its optimal value and associated LMP
	trajectory. 
	
	The next subsection recovers the global ESA--ISO Stackelberg
	equilibrium by comparing the optimal values across critical regions.}

\subsection{ESA--ISO Stackelberg Equilibrium}
\label{subsec:global_esa_decomposition}

\revv{Having characterized the ESA problem within each critical region, we now
	recover the global solution to the ESA problem and the resulting ESA--ISO
	Stackelberg equilibrium by comparing the optimal values of the critical-region
	subproblems, as stated in the following theorem.}

\revv{\begin{theorem}[Stackelberg equilibrium via critical-region decomposition]
	\label{thm:critical_region_decomposition}
	Under Assumptions~\ref{assum1}--\ref{assum3}, any ESA schedule solving
	\eqref{best:ESA:nonrestrict} is globally optimal for~\eqref{bilevel} and,
	together with its corresponding ISO market-clearing response, yields an
	ESA--ISO Stackelberg equilibrium:
	\begin{equation}
		\min_{r\in\mathcal R}
		\;
		\min_{\substack{
				\mathbf u\in\mathcal U(\mathbf B)\\
				\mathbf F_r\operatorname{vec}(\mathbf u)
				\leq
				\mathbf h_r
		}}
		J_r(\mathbf u).
		\label{best:ESA:nonrestrict}
	\end{equation}
\end{theorem}
}

Theorem~\ref{thm:critical_region_decomposition} turns the ESA's strategic
market-participation problem into an exact and transparent comparison across
critical regions and their associated constraint-binding patterns. Each
critical-region subproblem determines the ESA's best attainable objective value
and corresponding LMP trajectory. Comparing these optima ranks the attainable
critical regions according to the ESA's preference and reveals its preferred
binding patterns, while the associated objective terms explain why the ESA may
prefer one pattern over another.

\begin{remark}[ESA-reachable critical regions]
	\label{rem:reachable_critical_regions}
	Computing~\eqref{best:ESA:nonrestrict} requires identifying and solving
	only the critical regions that intersect $\mathcal U(\mathbf B)$, since
	the remaining regions contain no feasible ESA schedule. Thus, even when
	$\mathcal U_{\mathrm{mc}}$ contains many critical regions, the ESA-reachable subset
	may be much smaller. The on-demand algorithm presented in Appendix~\ref{app:reachable_region_exploration} discovers and
	constructs these regions as needed, without first constructing the full
	critical-region decomposition of $\mathcal U_{\mathrm{mc}}$ or enumerating all
	potential intersections with $\mathcal U(\mathbf B)$.
\end{remark}

\section{Constraint-Binding-Pattern Manipulation and FTR Effects}
\label{sec:pattern_manipulation_ftr}

\revv{The preceding section characterized how the ESA may increase its profit by
inducing particular constraint-binding patterns. We now formally define
constraint-binding-pattern manipulation and examine how FTR positions may
alter the ESA's preferred pattern and the associated welfare outcome.}

\subsection{Constraint-Binding-Pattern Manipulation}
\label{subsec:manipulation_benchmark}

\revv{The constraint-binding pattern induced by an ESA schedule solving
	\eqref{best:ESA:nonrestrict} does not by itself establish manipulation. Storage
	participation can alter the binding pattern relative to the network without
	storage even under nonstrategic operation. To isolate changes attributable to
	the ESA's strategic price-making behavior, we therefore require a reference
	that preserves the same storage capacity while removing strategic price effects.
	We use a \emph{centrally dispatched-storage benchmark}, in which the ISO jointly
	dispatches storage, generation, and demand to maximize social welfare, subject
	to the same bus-specific storage-capacity vector $\mathbf B$ as the strategic
	ESA:}

\begin{subequations}
	\label{eq:centralized_storage_benchmark}
	\begin{align}
		\mathrm{SW}^{\mathrm{cen}}(\mathbf B)
		:=
		\max_{\mathbf u,\,\mathbf g,\,\mathbf d,\,\mathbf p}
		\quad
		&
		\sum_{t\in\mathcal T}
		\left[
		V_t\!\left(\mathbf d(t)\right)
		-
		C_t\!\left(\mathbf g(t)\right)
		\right]
		\label{eq:centralized_storage_benchmark_obj}
		\\
		\text{s.t.}\quad
		&
		\eqref{DCOPF:net_power:full},
		~
		\eqref{DCOPF:Constraints},
		\label{eq:centralized_storage_benchmark_market}
		\\
		&
		\mathbf u\in\mathcal U(\mathbf B).
		\label{eq:centralized_storage_benchmark_storage}
	\end{align}
\end{subequations}

\revv{Although the storage schedule in
	\eqref{eq:centralized_storage_benchmark} is determined through centralized
	dispatch, the following proposition shows that the same schedule is also
	optimal for a price-taking ESA that treats the associated LMPs as given and
	does not strategically account for its effect on prices.}

\begin{proposition}[Price-taking interpretation of centralized storage dispatch]
	\label{prop:price_taking_implementation}
	For any optimal solution of
	\eqref{eq:centralized_storage_benchmark}, let
	$\mathbf u^{\mathrm{cen}}$ denote its storage schedule and
	$\bm\lambda^{\mathrm{cen}}$ the associated LMP trajectory. Then, for a price-taking ESA facing these LMPs,
	\begin{equation}
		\mathbf u^{\mathrm{cen}}
		\!\in
		\!\!\operatorname*{arg\,min}_{\mathbf u\in
			\mathcal U(\mathbf B)\cap\mathcal U_{\mathrm{mc}}}
		\!
		\sum_{t\in\mathcal T}
		\bm\lambda^{\mathrm{cen}}(t)^\top\mathbf u(t)
		\!-\!\!\!\!
		\sum_{(i,j)\in\mathcal F}
		\!\!\!\bm\kappa_{(i,j)}^\top
		\!\left(
		\bm\lambda_j^{\mathrm{cen}}
		-
		\bm\lambda_i^{\mathrm{cen}}
		\right)
		\!.
		\label{eq:price_taking_implementation}
	\end{equation}
\end{proposition}

\revv{Proposition~\ref{prop:price_taking_implementation} shows that the
	centralized-optimal storage schedule is also an
	optimal choice for a price-taking ESA facing the same LMPs. Thus, the benchmark has an economic interpretation consistent with a price-taking ESA without strategic price impact, providing a natural reference against which to identify
	constraint-binding-pattern manipulation.}

\revv{Let $\mathcal U_{\mathrm{cen}}^\star(\mathbf B)$ denote the storage
	schedules appearing in optimal solutions of
	\eqref{eq:centralized_storage_benchmark}. Because $\mathbf u$ does not enter
	the objective directly, this set may contain multiple schedules. Although
	generation and demand are unique under Assumption~\ref{assum1}, different
	centralized-optimal storage schedules may change net injections and line flows
	and thus induce different active sets and critical regions. We therefore
	compare the binding pattern arising from the strategic ESA equilibrium
	in~\eqref{best:ESA:nonrestrict} with the patterns induced by all
	centralized-optimal storage schedules rather than selecting a single benchmark
	solution.}
	
\revv{To perform this comparison, we need to specify which active constraints
	should determine whether the strategic-equilibrium and benchmark patterns match. The Lagrange multiplier of an active operating limit measures the
	marginal change in the ISO's optimal objective value resulting from a relaxation
	of that limit. Hence, an active inequality with a zero multiplier is binding in
	the primal solution but has no marginal economic value to the ISO. Using such
	constraints to distinguish patterns could therefore classify two outcomes as
	different solely because of a weakly active limit with no marginal effect on
	the market-clearing objective. To exclude such distinctions, define}
\begin{equation}
	\mathcal A^+(\mathbf u)
	:=
	\left\{
	m\in\mathcal A(\mathbf u):
	\mu_m^\star(\mathbf u)>0
	\right\},
	\label{eq:positive_multiplier_binding_pattern}
\end{equation}
\revv{where $m$ indexes the operating-limit inequalities in~\eqref{DCOPF}, and
	$\mu_m^\star(\mathbf u)$ is the corresponding optimal Lagrange multiplier.
	Using~\eqref{eq:positive_multiplier_binding_pattern}, define the set of
	positive-multiplier binding patterns induced by centralized-optimal storage
	schedules as}
\begin{equation}
	\mathcal K_{\mathrm{cen}}^+(\mathbf B)
	:=
	\left\{
	\mathcal A^+(\mathbf u):
	\mathbf u\in\mathcal U_{\mathrm{cen}}^\star(\mathbf B)
	\right\}.\nonumber
\end{equation}
\revv{Thus, $\mathcal K_{\mathrm{cen}}^+(\mathbf B)$ collects the benchmark
	binding patterns after excluding weakly active constraints. We can now
	formally compare the positive-multiplier binding pattern arising from the
	strategic ESA equilibrium with this benchmark set.}

\begin{definition}[Constraint-binding-pattern manipulation]
	\label{def:binding_pattern_manipulation}
	An equilibrium ESA schedule $\mathbf u^\star$ obtained from
	\eqref{best:ESA:nonrestrict} exhibits
	\emph{constraint-binding-pattern manipulation} relative to the centrally
	dispatched-storage benchmark if
	\begin{equation}
		\mathcal A^+(\mathbf u^\star)
		\notin
		\mathcal K_{\mathrm{cen}}^+(\mathbf B).\nonumber
	\end{equation}
\end{definition}

\revv{Thus, the equilibrium exhibits manipulation when its
	positive-multiplier binding pattern matches none of the centralized-optimal
	positive-multiplier patterns.}

\revv{Having characterized the ESA's incentives to induce particular
	constraint-binding patterns and formalized when such behavior constitutes
	manipulation, we now examine how FTR positions may alter the ESA's preferred
	pattern and the resulting social welfare outcome.}

\subsection{FTR Effects on Pattern Selection and Social Welfare}
\label{subsec:ftr_welfare_implications}

\revv{FTR positions can alter the ESA's ranking of attainable
	constraint-binding patterns by modifying the linear and constant terms of the
	objective within each critical region in~\eqref{eq:regional_linear_term}--\eqref{eq:regional_constant_term}.
		A pattern that produces an LMP spread aligned with an FTR position increases
		the associated FTR payoff, whereas an opposing spread reduces it and may make
		the payoff negative. Thus, FTR positions can change the relative
		profitability of attainable patterns and redirect the ESA's preferred pattern.}
	
\revv{To assess the welfare implications of such FTR-induced changes in pattern
	selection, define the social welfare associated with an ESA schedule $\mathbf u$
	as}
\begin{equation}
	\mathrm{SW}(\mathbf u)
	:=
	\sum_{t\in\mathcal T}
	\left[
	V_t\!\left(\mathbf d^\star(t;\mathbf u)\right)
	-
	C_t\!\left(\mathbf g^\star(t;\mathbf u)\right)
	\right].
	\label{eq:social_welfare}
\end{equation}
	
	\revv{In addition to the centrally dispatched-storage benchmark
		$\mathrm{SW}^{\mathrm{cen}}(\mathbf B)$ defined in
		Subsection~\ref{subsec:manipulation_benchmark}, we use
		$\mathrm{SW}^{\mathrm{cen}}(\mathbf 0)$ as the \emph{no-storage benchmark}.
	The following theorem compares $\mathrm{SW}(\mathbf u^\star)$ with these two
	benchmarks, first without FTR positions and then with nonzero FTR positions.}

\begin{theorem}[Welfare implications of FTR positions]
	\label{thm:ftr_welfare_implications}
Consider the ESA--ISO Stackelberg interaction~\eqref{bilevel} for a fixed
storage-capacity vector $\mathbf B$, assume
$\mathbf 0\in\mathcal U_{\mathrm{mc}}$, and let $\mathbf u^\star$ denote an
equilibrium ESA schedule obtained from~\eqref{best:ESA:nonrestrict}. Then the following hold:
	\begin{enumerate}
		\item In the absence of FTR positions, i.e.,
		$\bm{\kappa}=\mathbf 0$,
		\begin{equation}
			\mathrm{SW}^{\mathrm{cen}}(\mathbf B)
			\geq
			\mathrm{SW}(\mathbf u^\star)
			\geq
			\mathrm{SW}^{\mathrm{cen}}(\mathbf 0).
			\label{eq:welfare_without_ftr}
		\end{equation}
		
		\item There exist instances with $\bm{\kappa}\neq\mathbf 0$ for which
		\begin{equation}
			\mathrm{SW}^{\mathrm{cen}}(\mathbf B)
			\geq
			\mathrm{SW}^{\mathrm{cen}}(\mathbf 0)
			>
			\mathrm{SW}(\mathbf u^\star).
			\label{eq:welfare_with_ftr}
		\end{equation}
	\end{enumerate}
\end{theorem}

\revv{Theorem~\ref{thm:ftr_welfare_implications} separates the welfare benefit
	of available storage capacity from the effects of its strategic operation.
	Without FTR positions, strategically operated storage weakly improves social
	welfare relative to the no-storage benchmark, while centralized dispatch
	provides an upper bound, consistent with prior findings such
	as~\cite{contreras2017participation}. With FTR positions, however, the
	no-storage lower bound can fail. By altering the ESA's ranking of attainable
	binding patterns through congestion-dependent FTR payoffs, FTR positions can
	strengthen the ESA's incentive to induce a pattern that yields lower welfare.
	This result highlights the need for market rules that jointly account for
	physical market participation and financial positions such as FTRs when
	assessing market participants' strategic behavior.}

\section{ISO Tools for Limiting Constraint-Binding-Pattern Manipulation}
\label{sec:iso_binding_pattern_control}

\revv{Having characterized the ESA's incentives to induce particular
	constraint-binding patterns, we now develop two ISO-side tools for limiting
	undesirable outcomes. Suppose the ISO identifies a set of desirable patterns
	based on system-operation or market criteria. A natural starting point is the
	set of patterns induced by the centrally dispatched-storage benchmark. The ISO
	may broaden this set by admitting additional constraint-binding patterns, for
	example, those that avoid congestion on specified transmission lines.
	Although some additionally permitted patterns may still constitute
	constraint-binding-pattern manipulation under
	Definition~\ref{def:binding_pattern_manipulation}, they are nevertheless
	acceptable under the ISO's criteria. Admitting such patterns provides the ESA with greater scheduling flexibility while limiting the range of undesirable patterns it can induce. Given this
	set, one tool restricts the admissible schedules of an ESA with a given capacity
	vector $\mathbf B$, whereas the other imposes bus-specific limits on the storage
	capacity that the ESA may aggregate. Both tools preserve the ESA's freedom to
	choose among schedules that induce permitted patterns, without requiring it to
	follow a prescribed schedule.}

Throughout this section, let
$
\mathcal R_{\mathrm{des}}
\subseteq
\mathcal R
$
denote the indices of critical regions whose relative interiors%
\footnote{\textcolor{black}{Critical regions are closed, so boundaries shared by approved and
		unapproved neighboring regions also belong to
		$\mathcal C_{\mathrm{des}}$. This boundary overlap does not alter the
		market-outcome interpretation because constraints active solely on such
		boundaries are weakly active, with zero Lagrange multipliers, and therefore
		have no marginal effect on the market-clearing objective.}}
are deemed desirable and therefore approved by the ISO, and define
$
\mathcal C_{\mathrm{des}}
:=
\bigcup_{r\in\mathcal R_{\mathrm{des}}}
\mathcal C_r.
$
To ensure feasibility under both tools, we assume that the no-action schedule
$\mathbf u=\mathbf 0$ is approved:
\begin{equation}
	\mathbf 0
	\in
	\mathcal C_{\mathrm{des}}.
	\label{eq:zero_schedule_approved}
\end{equation}

Since
$
\mathbf 0
\in
\mathcal U(\mathbf B)
$
for every
$
\mathbf B
\geq
\mathbf 0,
$
this condition ensures that the ESA retains at least one admissible schedule.

\subsection{Schedule Restrictions for a Given Storage-Capacity Vector}
\label{subsec:region_restricted_given_B}

We first consider an ESA participating with a given storage-capacity vector
$\mathbf B$. Existing market rules already impose quantity-side requirements
on market participants through mechanisms such as must-offer obligations and rules addressing physical withholding~\cite{isoNEMustOffer}. Motivated by this practice, we develop a schedule-based
tool tailored to strategic ESAs. Rather than imposing separate limits on
individual charging and discharging quantities, the proposed tool restricts the
ESA's admissible multi-period and multi-location schedules based on the
constraint-binding patterns they induce.

Knowing $\mathcal C_{\mathrm{des}}$, the ISO can restrict the ESA's admissible
schedules to
$
\mathcal U(\mathbf B)
\cap
\mathcal C_{\mathrm{des}}
$. The resulting schedule-restricted ESA problem is
\begin{equation}
	\min_{\mathbf u\in
		\mathcal U(\mathbf B)
		\cap
		\mathcal C_{\mathrm{des}}}
	J(\mathbf u).
	\label{eq:region_restricted_esa_problem}
\end{equation}

Using the critical-region decomposition developed in
Subsection~\ref{subsec:regional_esa_problems},
problem~\eqref{eq:region_restricted_esa_problem} can be solved by solving
the critical-region subproblem~\eqref{eq:regional_esa_qp} for each
$r\in\mathcal R_{\mathrm{des}}$ and then comparing their optimal values, as in
\eqref{best:ESA:nonrestrict}, with the comparison restricted to the approved
regions.

\subsection{Limits on the Aggregated Storage-Capacity Vector}
\label{subsec:capacity_limits_binding_pattern_control}

We next consider an alternative tool based on aggregation-capacity limits.
Existing market rules already impose limits on aggregation size to manage
network impacts. For example, CAISO limits a distributed energy resource
aggregation spanning multiple pricing nodes to 20~MW~\cite{caiso2026tariff}.
Building on this principle, we develop a more targeted, network-aware approach.
Rather than imposing a single limit on the ESA's total aggregated storage
capacity, the ISO sets a bus-specific maximum capacity $\overline B_i$ at each
bus $i$ for market participation. The ISO seeks a capacity-limit vector
$\overline{\mathbf B}$ satisfying the approved-set containment condition
\begin{equation}
	\mathcal U\!\left(\overline{\mathbf B}\right)
	\subseteq
	\mathcal C_{\mathrm{des}}.
	\label{eq:desired_set_containment}
\end{equation}

By~\eqref{eq:zero_schedule_approved},
$
\overline{\mathbf B}
=
\mathbf 0
$
always satisfies~\eqref{eq:desired_set_containment}, since
$
\mathcal U(\mathbf 0)
=
\{\mathbf 0\}
\subseteq
\mathcal C_{\mathrm{des}}.
$
Thus, a feasible capacity-limit vector always exists, although the amount of
positive storage capacity that can be safely admitted depends on $\mathcal C_{\mathrm{des}}$.
Because $\mathcal C_{\mathrm{des}}$ may consist of multiple critical regions,
it need not be convex. We first consider the case in which it admits a single
polyhedral representation, yielding explicit linear capacity-limit conditions,
and then discuss the general nonconvex case.

\subsubsection{Polyhedral approved set}
\label{subsubsec:polyhedral_approved_sets}

For this case, we adopt the following assumption on the geometry of the
approved set.
\begin{assumption}[Polyhedral representation of the approved set]
	\label{assump:polyhedral_desired_set}
	The ISO-approved set $\mathcal C_{\mathrm{des}}$ admits the representation
	\begin{equation}
		\mathcal C_{\mathrm{des}}
		=
		\left\{
		\mathbf u:
		\mathbf F_{\mathrm{des}}
		\operatorname{vec}(\mathbf u)
		\leq
		\mathbf h_{\mathrm{des}}
		\right\},
		\label{eq:desired_binding_pattern_set}
	\end{equation}
	where
	$
	\mathbf F_{\mathrm{des}}
	\in
	\mathbb R^{m_{\mathrm{des}}\times|\mathcal N||\mathcal T|}
	$
	and
	$
	\mathbf h_{\mathrm{des}}
	\in
	\mathbb R^{m_{\mathrm{des}}},
	$
	with $m_{\mathrm{des}}$ denoting the number of inequalities defining the
	approved set.
\end{assumption}

Using~\eqref{eq:desired_binding_pattern_set},
condition~\eqref{eq:desired_set_containment} is equivalent to
\begin{equation}
	\mathbf F_{\mathrm{des}}
	\operatorname{vec}(\mathbf u)
	\leq
	\mathbf h_{\mathrm{des}},
	\qquad
	\forall\mathbf u
	\in
	\mathcal U\!\left(\overline{\mathbf B}\right),
	\label{eq:universal_desired_conditions}
\end{equation}
which is stated over feasible schedules, but capacity-limit design requires a
condition directly in $\overline{\mathbf B}$. The following proposition
provides this characterization. To state it, let $\mathbf f_k^\top$ denote the
$k$-th row of $\mathbf F_{\mathrm{des}}$ and partition it by bus as
\begin{equation}
	\mathbf f_k^\top
	=
	\bigl[
	\mathbf f_{k,1}^\top
	~
	\cdots
	~
	\mathbf f_{k,|\mathcal N|}^\top
	\bigr],
	\qquad
	k=1,\ldots,m_{\mathrm{des}},
	\label{eq:desired_row_partition}
\end{equation}
where
$
\mathbf f_{k,i}\in\mathbb R^{|\mathcal T|}.
$

\begin{proposition}[Linear characterization of approved-set containment]
	\label{prop:linear_safe_capacity_limits}
	Under Assumption~\ref{assump:polyhedral_desired_set}, for each
	$
	k=1,\ldots,m_{\mathrm{des}}
	$
	and
	$
	i\in\mathcal N,
	$
	define
		\begin{align}
			\xi_{k,i}
			:=
			\max_{\mathbf w_i\in\mathbb R^{|\mathcal T|}}
			\quad
			&
			\mathbf f_{k,i}^\top
			\mathbf w_i\nonumber
			\\
			\text{s.t.}\quad \,\,\,\,
			&
			\mathbf 0
			\leq
			\mathbf L\mathbf w_i
			\leq
			\mathbf 1,\nonumber
			\\
			&
			\mathbf 1^\top
			\mathbf w_i
			=
			0.\nonumber
		\end{align}
	Let
	$
	\bm{\Xi}
	\in
	\mathbb R^{m_{\mathrm{des}}\times|\mathcal N|}
	$
	denote the matrix whose $(k,i)$-th entry is $\xi_{k,i}$. Then, for any
	$
	\overline{\mathbf B}
	\geq
	\mathbf 0,
	$
	\begin{equation}
		\mathcal U\!\left(\overline{\mathbf B}\right)
		\subseteq
		\mathcal C_{\mathrm{des}}
		\quad
		\Longleftrightarrow
		\quad
		\bm{\Xi}
		\overline{\mathbf B}
		\leq
		\mathbf h_{\mathrm{des}}.
		\label{eq:capacity_limit_linear_characterization}
	\end{equation}
\end{proposition}

 Equipped with Proposition~\ref{prop:linear_safe_capacity_limits}, the ISO can
 replace~\eqref{eq:universal_desired_conditions} with
 $m_{\mathrm{des}}$ linear inequalities in $\overline{\mathbf B}$.
Here, each coefficient $\xi_{k,i}$ quantifies the maximum contribution of one unit
 of storage capacity at bus $i$ to the $k$-th approved-set inequality, thereby
 linking bus-specific capacity directly to pattern admissibility.

There may be multiple capacity-limit vectors satisfying
\eqref{eq:capacity_limit_linear_characterization}. A natural policy is to
maximize the total admitted capacity across the network, formulated as
\begin{subequations}
	\label{eq:capacity_limit_design_lp}
	\begin{align}
		\max_{\overline{\mathbf B}}\quad
		&
		\sum_{i\in\mathcal N}\overline{B}_i
		\label{eq:capacity_limit_design_lp_obj}
		\\
		\text{s.t.}\quad
		&
		\bm{\Xi}\overline{\mathbf B}
		\leq
		\mathbf h_{\mathrm{des}},
		\label{eq:capacity_limit_design_lp_containment}
		\\
		&
		\overline{\mathbf B}
		\geq
		\mathbf 0.
		\label{eq:capacity_limit_design_lp_nonnegative}
	\end{align}
\end{subequations}

\begin{corollary}[Approved-pattern guarantee under capacity limits]
	\label{cor:safe_participation_capacity_limits}
	Under Assumption~\ref{assump:polyhedral_desired_set}, let
	$\overline{\mathbf B}$ be any feasible solution of
	\eqref{eq:capacity_limit_design_lp}. Then, for every aggregated
	storage-capacity vector $\mathbf B$ satisfying
	$
	\mathbf 0
	\leq
	\mathbf B
	\leq
	\overline{\mathbf B}
	$
	componentwise, it holds that
	\begin{equation}
		\mathcal U(\mathbf B)
		\subseteq
		\mathcal U\!\left(\overline{\mathbf B}\right)
		\subseteq
		\mathcal C_{\mathrm{des}}.
		\label{eq:safe_capacity_guarantee}
	\end{equation}
\end{corollary}

Corollary~\ref{cor:safe_participation_capacity_limits} ensures that an ESA
whose aggregated storage-capacity vector remains below $\overline{\mathbf B}$
cannot submit a schedule outside the ISO-approved set
$\mathcal C_{\mathrm{des}}$. If
\eqref{eq:capacity_limit_design_lp} is bounded, an optimal solution provides
this guarantee while admitting the maximum total storage capacity.

\subsubsection{Nonconvex approved set}
\label{subsubsec:nonconvex_approved_sets}

\revv{If Assumption~\ref{assump:polyhedral_desired_set} does not hold,
	$\mathcal C_{\mathrm{des}}$ may be a nonconvex union of critical regions.
	A conservative approach is to construct polyhedral inner approximations
	$
	\{\mathcal P_{\mathrm{des},\rho}\}_{\rho=1}^{R_{\mathrm{in}}}
	$
	such that
	\[
\mathbf 0
\in
\mathcal P_{\mathrm{des},\rho}
\subseteq
\mathcal C_{\mathrm{des}},
\qquad
\forall \rho=1,\ldots,R_{\mathrm{in}}.
	\]

	For each $\rho$, the corresponding inner approximation can be represented as
$
		\mathcal P_{\mathrm{des},\rho}
		=
		\left\{
		\mathbf u:
		\mathbf F_{\mathrm{in},\rho}
		\operatorname{vec}(\mathbf u)
		\leq
		\mathbf h_{\mathrm{in},\rho}
		\right\},
$
	for some matrix $\mathbf F_{\mathrm{in},\rho}$ and vector
	$\mathbf h_{\mathrm{in},\rho}$.
	Applying Proposition~\ref{prop:linear_safe_capacity_limits} to each
	$\mathcal P_{\mathrm{des},\rho}$ yields the corresponding admissible-capacity set
\begin{equation}
		\mathcal B_{\mathrm{in},\rho}
		:=
		\left\{
		\overline{\mathbf B}\geq\mathbf 0:
		\mathcal U(\overline{\mathbf B})
		\subseteq
		\mathcal P_{\mathrm{des},\rho}
		\right\}.\nonumber
\end{equation}
	Hence, $\mathcal B_{\mathrm{in}}:=\bigcup_{\rho=1}^{R_{\mathrm{in}}}
	\mathcal B_{\mathrm{in},\rho}$ is a conservative set of capacity-limit vectors
	that the ISO may announce to the ESA, allowing it to choose any vector in
	$\mathcal B_{\mathrm{in}}$. Alternatively, if the ISO wishes to retain the
	maximum-capacity objective in~\eqref{eq:capacity_limit_design_lp}, it may solve
	the corresponding problem for each $\rho$. If these problems admit optimizers,
	let $\overline{\mathbf B}_{\rho}^{\star}$ denote an optimizer for each $\rho$.
	The ISO may then announce the set
	\begin{equation}
		\mathcal B_{\mathrm{in}}^\star
		:=
		\left\{
		\overline{\mathbf B}_{\rho}^{\star}:
		\rho=1,\ldots,R_{\mathrm{in}}
		\right\}
		\subseteq
		\mathcal B_{\mathrm{in}},\nonumber
	\end{equation}
	from which the ESA may select a capacity-limit vector.}

\section{Numerical Experiments}\label{numerical}

To keep the exposition and critical-region geometry simple and transparent, we
first illustrate the main ideas using a three-bus example and then demonstrate
the scalability of the framework on several IEEE standard test systems.

\subsection{Illustrative Three-Bus Example}\label{three:bus:example}

We consider the two-period system illustrated in
Fig.~\ref{fig:three_bus_system} and use the proposed constraint-binding-pattern
decomposition to obtain the ESA--ISO equilibrium. We compare three cases: the
centrally dispatched-storage benchmark, a strategic ESA without FTRs, and a
strategic ESA holding an FTR. The resulting outcomes are summarized in
Table~\ref{tab:three_bus_results}, while Fig.~\ref{fig:three_bus_regions}
shows the corresponding critical-region partition. In the table, the SW gap is defined as
$
\mathrm{SW}(\mathbf u)-\mathrm{SW}^{\mathrm{cen}}(\mathbf B).
$

\begin{figure}[h]
	\centering
	\begin{tikzpicture}[
		font=\scriptsize,
		>=Latex,
		bus1/.style={
			circle,
			draw=blue,
			fill=blue!5,
			minimum size=13mm,
			align=center
		},
		bus2/.style={
			circle,
			draw=orange!90!black,
			fill=orange!8,
			minimum size=13mm,
			align=center
		},
		bus3/.style={
			circle,
			draw=green!60!black,
			fill=green!5,
			minimum size=13mm,
			align=center
		},
		gen/.style={
			rectangle,
			draw=blue!70!black,
			fill=blue!5,
			rounded corners=1pt,
			align=center,
			inner sep=2pt
		},
		load/.style={
			rectangle,
			draw=gray!70!black,
			fill=gray!10,
			rounded corners=1pt,
			align=center,
			inner sep=2pt
		},
		storage/.style={
			rectangle,
			draw=magenta!70!black,
			fill=magenta!8,
			rounded corners=1pt,
			align=center,
			inner sep=2pt
		}
		]
		
		\node[bus1] (b1) at (0,0) {Bus 1};
		\node[bus2] (b2) at (2.45,0) {Bus 2};
		\node[bus3] (b3) at (4.90,0) {Bus 3};
		
		\draw[thick,->]
		(b1) --
		node[above] {$\overline f_{12}=300$}
		(b2);
		
		\draw[thick,->]
		(b2) --
		node[above] {$\overline f_{23}=200$}
		(b3);
		
		\node[gen,above=5mm of b1] (g1) {
			$C_1(g_{1})=0.05g_{1}^2+20g_{1}$\\
			$0\leq g_1\leq290$
		};
		\draw[->] (g1) -- (b1);
		
		\node[gen,above=5mm of b3] (g3) {
			$C_3(g_{3})=0.20g_{3}^2+40g_{3}$\\
			$0\leq g_3\leq200$
		};
		\draw[->] (g3) -- (b3);
		
		\node[load,below=5mm of b2] (d2) {
			$\mathbf d_2=(80,100)$
		};
		\draw[->] (b2) -- (d2);
		
		\node[load,below=5mm of b3] (d3) {
			$\mathbf d_3=(160,240)$
		};
		\draw[->] (b3) -- (d3);
		
		\node[storage,right=3mm of b3] (s3) {
			Storage\\
			$B_3=120$
		};
		\draw[<->] (b3) -- (s3);
		
\draw[
green!60!black,
dashed,
->,
out=22,
in=158
]
([xshift=1mm,yshift=-0.25mm]b1.north east) to
node[above=-1pt] {200-MW FTR at $t=1$}
([xshift=-1mm,yshift=-0.25mm]b3.north west);
		
	\end{tikzpicture}
	\caption{Illustrative two-period three-bus system. Each period is one hour; all
		power quantities shown are in MW, while storage capacity is in MWh. The shown
		generation-cost functions and transmission-line limits apply in both periods,
		and $\mathbf d_i=(d_i(1),d_i(2))$ gives the inelastic loads at bus $i$.}
	\label{fig:three_bus_system}
\end{figure}
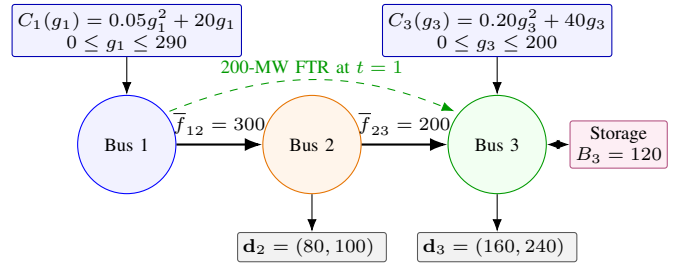

\begin{table}[h]
	\centering
\caption{Three-bus outcomes, rounded to one decimal place}
	\label{tab:three_bus_results}
	\small
	\setlength{\tabcolsep}{2.5pt}
	\begin{tabular}{@{}ccccc@{}}
		\hline
		\textbf{Case} &
		\textbf{$\mathbf u_3$} &
		\textbf{Region} &
		\textbf{$\bm\lambda_3$} &
		\textbf{SW gap (\$)} \\
		\hline
		Centralized &
		$(50,-50)$ &
		$\mathcal C_2$ &
		$(47.2,47.2)$ &
		$0.0$ \\
		
		No FTR &
		$(17.5,-17.5)$ &
		$\mathcal C_1$ &
		$(44.6,53.0)$ &
		$-100.5$ \\
		
		FTR &
		$(111.7,-111.7)$ &
		$\mathcal C_3$ &
		$(68.7,42.3)$ &
		$-731.3$ \\
		\hline
	\end{tabular}
\end{table}
The results show that the centrally dispatched-storage benchmark lies strictly
inside $\mathcal C_2$, where no market-clearing inequality binds. Notably, this coincidence is example-specific;
the absence of binding inequalities does not generally imply maximum social
welfare. Without FTRs, the strategic ESA selects an outcome in $\mathcal C_1$, where
$g_1(2)$ reaches its upper limit, thereby constituting constraint-binding-pattern
manipulation and resulting in a \$100.5 social-welfare loss. The FTR further changes the ESA's preferred binding pattern:
congestion of line $(2,3)$ at $t=1$ creates a favorable FTR payoff, shifting
the equilibrium to the line-congested region $\mathcal C_3$ and increasing
the welfare loss to \$731.3. Thus, even in this simple case, a strategic ESA
can alter which operating limits become binding, highlighting
constraint-binding-pattern manipulation as a mechanism affecting market outcomes
and the role of FTRs in further shifting the ESA toward different binding
patterns, with substantially lower social welfare in this example.

These results also illustrate how the proposed ISO tools can restrict
undesirable binding patterns without prescribing the ESA's exact schedule.
In the FTR case, for example, the ISO can exclude the relative interior of
the line-congested region $\mathcal C_3$ by imposing $u_3(1)\leq60$ or, equivalently,
$\overline B_3\leq60$. A tighter limit, $u_3(1)\leq27.5$, confines the ESA
to $\mathcal C_1$, where it selects $(17.5,-17.5)$. Alternatively, restricting the schedule to
$27.5\leq u_3(1)\leq60$ confines the ESA to the closed region
$\mathcal C_2$, yielding $(27.5,-27.5)$ and reducing the welfare loss to
\$40.5, compared with \$731.3 under the unrestricted FTR outcome. Thus, the proposed ISO tools can substantially improve the market outcome while
preserving the ESA's freedom to choose within the admissible regions.

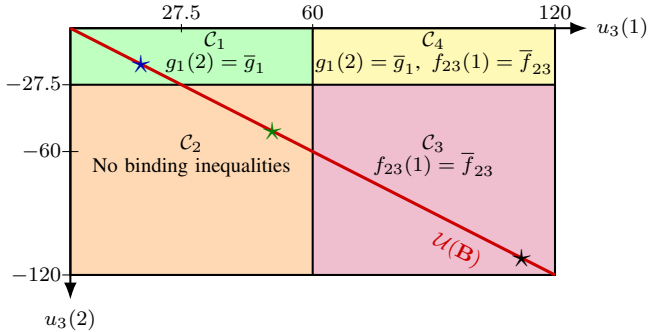
\begin{figure}[t]
	\centering
	\resizebox{\columnwidth}{!}{%
		\begin{tikzpicture}[
			x=0.055cm,
			y=0.028cm,
			font=\footnotesize,
			>=Latex
			]
			
			\fill[green!25]  (0,-27.5) rectangle (60,0);
			\fill[orange!30] (0,-120) rectangle (60,-27.5);
			\fill[purple!25] (60,-120) rectangle (120,-27.5);
			\fill[yellow!35] (60,-27.5) rectangle (120,0);
			
			\draw[thick] (0,0) -- (120,0) -- (120,-120) -- (0,-120) -- cycle;
			\draw[thick] (60,0) -- (60,-120);
			\draw[thick] (0,-27.5) -- (120,-27.5);
			
			\draw[->,thick]
			(0,0) -- (128,0)
			node[right] {$u_3(1)$};
			
			\draw[->,thick]
			(0,0) -- (0,-132)
			node[below] {$u_3(2)$};
			
			\draw (27.5,2) -- (27.5,-2);
			\node[above] at (27.5,0) {$27.5$};
			
			\draw (60,2) -- (60,-2);
			\node[above] at (60,0) {$60$};
			
			\draw (120,2) -- (120,-2);
			\node[above] at (120,0) {$120$};
			
			\draw (-2,-27.5) -- (2,-27.5);
			\node[left] at (0,-27.5) {$-27.5$};
			
			\draw (-2,-60) -- (2,-60);
			\node[left] at (0,-60) {$-60$};
			
			\draw (-2,-120) -- (2,-120);
			\node[left] at (0,-120) {$-120$};
			
			\node[align=center] at (36,-12) {
				$\mathcal C_1$\\[-0.2mm]
				$g_1(2)=\overline g_1$
			};
			
			\node[align=center] at (30,-62) {
				$\mathcal C_2$\\[-0.1mm]
				No binding inequalities
			};
			
			\node[align=center] at (90,-62) {
				$\mathcal C_3$\\[-0.2mm]
				$f_{23}(1)=\overline f_{23}$
			};
			
			\node[align=center] at (90,-12) {
				$\mathcal C_4$\\[-0.14mm]
				$g_1(2)=\overline g_1$,\,
				$f_{23}(1)=\overline f_{23}$
			};
			
			\draw[
			red!80!black,
			very thick
			]
			(0,0) -- (120,-120)
			node[
			pos=0.82,
			sloped,
			below=1pt,
			red!80!black
			]
			{$\mathcal U(\mathbf B)$};
			
			\node[
			green!50!black,
			font=\Large
			]
			at (50,-50)
			{$\star$};
			
			\node[
			blue!70!black,
			font=\Large
			]
			at (17.5,-17.5)
			{$\star$};
			
			\node[
			black,
			font=\Large
			]
			at (111.7,-111.7)
			{$\star$};
			
		\end{tikzpicture}
	}
	\caption{Relevant portion of $\mathcal U_{\mathrm{mc}}$ partitioned into critical regions
		$\mathcal C_1$--$\mathcal C_4$. The green, blue, and black stars mark the
		centrally dispatched-storage, no-FTR, and FTR outcomes, respectively.}
	\label{fig:three_bus_regions}
\end{figure}

\subsection{Scalability of the Critical-Region Decomposition Method}

We evaluate the scalability of the proposed critical-region decomposition
method for obtaining the equilibrium of the ESA--ISO Stackelberg game on larger
test systems. Specifically, we consider the IEEE 5-, 14-, 30-, and 57-bus
systems over a 24-hour horizon. The original nodal loads are scaled using a
common daily profile to construct a 24-hour demand profile for each system, and
storage is installed at every bus, with total energy capacity corresponding to
$20\%$ of system peak demand over one hour. This penetration level is
motivated by recent ERCOT storage deployment~\cite{ERCOT2025MembershipMeeting}.

All experiments are implemented in MATLAB using Gurobi and conducted on a
2024 MacBook Air with an Apple M3 processor and 24~GB of memory. The solution
times reported in Table~\ref{tab:scalability_comparison} indicate that the
proposed method remains computationally practical for the benchmark systems
considered.

\begin{table}[h]
	\centering
\caption{Computation times on IEEE test systems.}
	\label{tab:scalability_comparison}
	\begin{tabular}{c c c c c}
		\hline
		\textbf{IEEE system} & 5-bus & 14-bus & 30-bus & 57-bus \\
		\textbf{Time (s)}    & 4.53  & 6.92   & 11.92  & 48.00  \\
		\hline
	\end{tabular}
\end{table}

\section{Concluding Remarks}
\label{sec:conclusion}

This paper studies the strategic market participation of a monopolistic ESA
coordinating geographically distributed storage units. By treating the ESA
schedule as a parameter of the ISO market-clearing problem, we decompose the
schedule space into critical regions associated with distinct
constraint-binding patterns. Within each region, we establish an affine mapping
from the ESA schedule to the resulting dispatch and LMPs, allowing the ESA--ISO
Stackelberg equilibrium to be recovered from a collection of convex pattern-specific subproblems. We further show that the ESA can strategically induce profit-favorable
binding patterns and that FTR positions can alter these incentives and the
resulting pattern selection. In particular, FTR positions can eliminate the
welfare guarantee that strategic storage without FTRs provides relative to the
no-storage benchmark. Based on these results, we develop ISO-side tools that
restrict admissible ESA schedules or storage capacities to limit undesirable
binding patterns. Future work may consider competition among ESAs,
uncertainty, and less conservative treatments of nonconvex ISO-approved sets.
\bibliographystyle{IEEEtran}
\bibliography{bib}
\appendices

\section{Constructing ESA-Reachable Critical Regions}
\label{app:reachable_region_exploration}

The critical-region decomposition in Sections~\ref{parametric} and
\ref{sec:esa_binding_pattern_decomposition} reformulates the ESA problem as a
collection of regional convex quadratic programs. However, constructing and
solving~\eqref{eq:regional_esa_qp} for every possible constraint-binding
pattern may be computationally demanding because the number of such patterns
grows combinatorially.

\subsection{Combinatorial Growth of Constraint-Binding Patterns}
\label{app:critical_region_complexity}

Let $N_{\mathrm g}$ and $N_{\mathrm d}$ denote, respectively, the numbers of
dispatchable generators and elastic demands with nontrivial operating ranges.
Fixed renewable injections and inelastic demands are excluded because they do
not generate alternative binding statuses.

At each $t\in\mathcal T$, every dispatchable generator and elastic demand can
be at its lower limit, in the interior, or at its upper limit. Similarly, each
transmission line can be congested in either direction or uncongested. Hence,
the maximum number of single-period constraint-binding patterns is
\begin{equation}
	R_t^{\max}
	:=
	3^{N_{\mathrm g}+N_{\mathrm d}+|\mathcal L|}.\nonumber
\end{equation}
A full-horizon pattern selects one such pattern for every
$t\in\mathcal T$. Therefore, by the Cartesian-product structure, the maximum
number of full-horizon constraint-binding patterns, denoted by $R^{\max}$, is
\begin{equation}
	R^{\max}
	=
	\prod_{t\in\mathcal T}R_t^{\max}
	=
	3^{|\mathcal T|
		\left(
		N_{\mathrm g}
		+
		N_{\mathrm d}
		+
		|\mathcal L|
		\right)}.\nonumber
\end{equation}

Let $\mathcal Q_{\mathrm{mc}}(t)$ denote the index set of the lower and upper
market-clearing inequalities associated with dispatchable generators,
elastic demands, and transmission lines at time $t$. Its cardinality is
\begin{equation}
	|\mathcal Q_{\mathrm{mc}}(t)|
	=
	2
	\left(
	N_{\mathrm g}
	+
	N_{\mathrm d}
	+
	|\mathcal L|
	\right).\nonumber
\end{equation}

The count $R^{\max}$ includes all status combinations induced by these
lower--upper inequality pairs before market-clearing feasibility and
optimality are imposed. Thus, not every pattern defines a nonempty critical
region. Moreover, because the storage constraints in
$\mathcal U(\mathbf B)$ couple ESA actions across time, not every combination
of nonempty single-period regions corresponds to an ESA-feasible
full-horizon schedule.

To obtain the optimal solution of the ESA decision-making problem
in~\eqref{best:ESA:nonrestrict} without explicitly constructing these
exponentially many full-horizon combinations, we design
Algorithm~\ref{alg:on_demand_reachable_regions}. By exploiting the
time-separable structure of the market-clearing problem and its associated
parametric mappings, the algorithm first constructs and validates the
critical regions independently for each time interval. It then jointly
selects an ESA-feasible sequence of hourly regions and the corresponding
storage actions through a convex MIQP that enforces the intertemporal storage
constraints.

\subsection{Construction of Hourly Critical Regions}
\label{app:on_demand_region_exploration}

Because the study intervals have unit duration and no separate charging or
discharging power limits are imposed, every hourly ESA action associated with
a schedule in $\mathcal U(\mathbf B)$ lies within
\begin{equation}
	\mathcal U_t^{\mathrm{box}}
	:=
	\left\{
	\mathbf u(t)\in\mathbb R^{|\mathcal N|}:
	-\mathbf B
	\leq
	\mathbf u(t)
	\leq
	\mathbf B
	\right\},\nonumber
\end{equation}
where the inequalities are interpreted componentwise. This box is an outer
approximation of the hourly action space. The exact intertemporal storage
constraints remain represented by $\mathbf u\in\mathcal U(\mathbf B)$ and
are imposed in the intertemporal region-selection problem.

For each $t\in\mathcal T$, a safe screening procedure is first applied over
$\mathcal U_t^{\mathrm{box}}$ to remove market-clearing inequalities that
cannot become binding within the hourly action range of interest. Let
\begin{equation}
	\mathcal Q_{\mathrm{scr}}(t)
	\subseteq
	\mathcal Q_{\mathrm{mc}}(t)\nonumber
\end{equation}
denote the set of inequalities retained after screening. The screening
procedure is safe in the sense that no removed inequality can be active at an
optimal hourly market-clearing solution for any
$\mathbf u(t)\in\mathcal U_t^{\mathrm{box}}$. Screening is used only to reduce
the active-set enumeration; all original market-clearing inequalities remain
enforced when each candidate critical region is constructed and validated.

Although LICQ is assumed to hold over $\mathcal U_{\mathrm{mc}}$ in the theoretical
development, the computational construction explicitly verifies LICQ for
each candidate active set and retains only LICQ-valid regions. This ensures
that the affine parametric mappings are constructed only for valid active
sets in the realized market instance.

To this end, let $n_x$ denote the dimension of the hourly lower-level decision
vector after excluding fixed renewable generation and inelastic demand, and let
$\mathbf A_{\mathrm{eq}}(t)$ denote the corresponding equality-constraint
coefficient matrix. The number of degrees of freedom remaining after imposing
the equality constraints is
\begin{equation}
	n_t^{\mathrm{dof}}
	:=
	n_x
	-
	\operatorname{rank}
	\left(
	\mathbf A_{\mathrm{eq}}(t)
	\right).\nonumber
\end{equation}
Under LICQ, no active set can contain more than
$n_t^{\mathrm{dof}}$ linearly independent inequality constraints. The
algorithm therefore enumerates physically admissible candidate active-index
sets satisfying
\begin{equation}
	\mathcal A_s(t)
	\subseteq
	\mathcal Q_{\mathrm{scr}}(t),
	\qquad
	\left|\mathcal A_s(t)\right|
	\leq
	n_t^{\mathrm{dof}},\nonumber
\end{equation}
where each index in $\mathcal A_s(t)$ identifies a retained inequality
assumed to be binding in candidate hourly region $s$ at time $t$. Candidates
containing both indices in the lower--upper bound pair of the same nonfixed
quantity are discarded.

To state the LICQ test, let $\mathbf a_j(t)^\top$ denote the coefficient row
of inequality $j\in\mathcal Q_{\mathrm{scr}}(t)$, and define
\begin{equation}
	\mathbf A_{\mathcal A_s(t)}(t)
	:=
	\begin{bmatrix}
		\mathbf a_j(t)^\top
	\end{bmatrix}_{j\in\mathcal A_s(t)},\nonumber
\end{equation}
as the matrix formed by stacking the coefficient rows of the inequalities
indexed by $\mathcal A_s(t)$. The candidate active set satisfies LICQ if
\begin{equation}
	\operatorname{rank}
	\left(
	\begin{bmatrix}
		\mathbf A_{\mathrm{eq}}(t)\\
		\mathbf A_{\mathcal A_s(t)}(t)
	\end{bmatrix}
	\right)
	=
	\operatorname{rank}
	\left(
	\mathbf A_{\mathrm{eq}}(t)
	\right)
	+
	\left|\mathcal A_s(t)\right|.\nonumber
\end{equation}

For each LICQ-valid candidate active set, the corresponding fixed-active-set
KKT system is constructed. Applying the argument of
Theorem~\ref{thm:affine_mapping} to the hourly market-clearing problem yields
the affine price mapping
\begin{equation}
	\bm{\lambda}^{\star}(t;\mathbf u)
	=
	\mathbf K_s(t)\mathbf u(t)
	+
	\mathbf k_s(t)\nonumber
\end{equation}
and the associated hourly critical region
\begin{equation}
	\mathcal C_{s,t}
	=
	\left\{
	\mathbf u(t)\in\mathbb R^{|\mathcal N|}:
	\mathbf F_{s,t}\mathbf u(t)
	\leq
	\mathbf h_{s,t}
	\right\}.\nonumber
\end{equation}

A candidate is retained only if
$\mathcal C_{s,t}\cap\mathcal U_t^{\mathrm{box}}$
admits a strictly positive interior margin relative to the hourly
storage-action space. Let $\mathcal R_t^{\mathrm h}$ denote the resulting
index set of validated hourly critical regions.

Using Lemma~\ref{lem:regional_esa_objective}, the contribution of
$s\in\mathcal R_t^{\mathrm h}$ to the ESA objective is
\begin{equation}
	J_{s,t}\!\left(\mathbf u(t)\right)
	=
	\mathbf u(t)^{\!\top}
	\mathbf K_s(t)
	\mathbf u(t)
	+
	\bm{\psi}_s(t)^{\!\top}
	\mathbf u(t)
	+
	\zeta_s(t),\nonumber
\end{equation}
where $\bm{\psi}_s(t)$ and $\zeta_s(t)$ are defined analogously to
\eqref{eq:regional_linear_term} and
\eqref{eq:regional_constant_term}, respectively, and
$\mathbf K_s(t)\succeq\mathbf 0$.

\subsection{Intertemporal Selection of Hourly Critical Regions}
\label{app:intertemporal_region_selection}

After constructing $\mathcal R_t^{\mathrm h}$ for every $t\in\mathcal T$,
the algorithm selects one region at each time interval. The possible
full-horizon sequences belong to
\begin{equation}
	\prod_{t\in\mathcal T}\mathcal R_t^{\mathrm h},\nonumber
\end{equation}
but this Cartesian product is not enumerated explicitly.

Let $z_s(t)\in\{0,1\}$ indicate whether region
$s\in\mathcal R_t^{\mathrm h}$ is selected at time $t$. Exactly one region is
selected at each interval:
\begin{equation}
	\sum_{s\in\mathcal R_t^{\mathrm h}}
	z_s(t)
	=
	1,
	\qquad
	t\in\mathcal T.
	\label{eq:select_one_hourly_region}
\end{equation}

For each $s\in\mathcal R_t^{\mathrm h}$, introduce an auxiliary
region-specific action
$\mathbf v_s(t)\in\mathbb R^{|\mathcal N|}$ and define the actual hourly ESA
action as
\begin{equation}
	\mathbf u(t)
	=
	\sum_{s\in\mathcal R_t^{\mathrm h}}
	\mathbf v_s(t),
	\qquad
	t\in\mathcal T.
	\label{eq:hourly_action_disaggregation}
\end{equation}
The selected critical-region constraints are imposed through
\begin{align}
	\mathbf F_{s,t}\mathbf v_s(t)
	&\leq
	\mathbf h_{s,t}z_s(t),
	\label{eq:scaled_hourly_region}\\
	-\mathbf Bz_s(t)
	\leq
	\mathbf v_s(t)
	&\leq
	\mathbf Bz_s(t),
	\label{eq:scaled_hourly_action}
\end{align}
for every $s\in\mathcal R_t^{\mathrm h}$ and $t\in\mathcal T$.

The resulting intertemporal region-selection problem is formulated as
\begin{equation}
	\begin{aligned}
		\min_{\substack{
				\mathbf u,\,
				\mathbf v_s(t),\,
				z_s(t)\\
				s\in\mathcal R_t^{\mathrm h},\,
				t\in\mathcal T}}
		\quad
		&
		\sum_{t\in\mathcal T}
		\sum_{s\in\mathcal R_t^{\mathrm h}}
		\left[
		\mathbf v_s(t)^{\!\top}
		\mathbf K_s(t)
		\mathbf v_s(t)
		+
		\bm{\psi}_s(t)^{\!\top}
		\mathbf v_s(t)
		\right]
		\\[-1mm]
		&
		+
		\sum_{t\in\mathcal T}
		\sum_{s\in\mathcal R_t^{\mathrm h}}
		\zeta_s(t)z_s(t)
		\\
		\mathrm{s.t.}\quad
		&
		\eqref{eq:select_one_hourly_region},
		\eqref{eq:hourly_action_disaggregation},
		\eqref{eq:scaled_hourly_region},
		\eqref{eq:scaled_hourly_action},
		\\
		&
		\mathbf u\in\mathcal U(\mathbf B),
		\\
		&
		z_s(t)\in\{0,1\}.
	\end{aligned}
	\label{eq:intertemporal_region_selection}
\end{equation}

Because $\mathbf K_s(t)\succeq\mathbf0$ for all
$s\in\mathcal R_t^{\mathrm h}$ and $t\in\mathcal T$, the continuous
quadratic objective in~\eqref{eq:intertemporal_region_selection} is convex,
and the resulting formulation is a convex MIQP. The binary variables
implicitly select an element of the Cartesian product of the hourly region
sets without explicitly enumerating all full-horizon combinations. At the
same time, $\mathbf u\in\mathcal U(\mathbf B)$ enforces the exact
intertemporal storage constraints and excludes combinations of hourly regions
that do not admit a feasible ESA schedule.

Algorithm~\ref{alg:on_demand_reachable_regions} summarizes the complete
procedure.

\begin{algorithm}[t]
	\caption{Hourly Critical-Region Construction and Intertemporal ESA
		Optimization}
	\label{alg:on_demand_reachable_regions}
	
	\KwIn{Storage-capacity vector $\mathbf B$ and the market-clearing
		problem~\eqref{DCOPF}.}
	
	\tcp{Stage 1: Construct the hourly critical-region sets}
	\For{each $t\in\mathcal T$}{
		Construct $\mathcal U_t^{\mathrm{box}}$\;
		
		Apply safe inequality screening and obtain
		$\mathcal Q_{\mathrm{scr}}(t)$\;
		
		Compute $n_t^{\mathrm{dof}}$ and set
		$\mathcal R_t^{\mathrm h}\leftarrow\varnothing$\;
		
		Enumerate all physically admissible active-index sets
		$\mathcal A_s(t)\subseteq\mathcal Q_{\mathrm{scr}}(t)$
		satisfying
		$\left|\mathcal A_s(t)\right|\leq n_t^{\mathrm{dof}}$\;
		
		\For{each candidate active-index set $\mathcal A_s(t)$}{
			Construct $\mathbf A_{\mathcal A_s(t)}(t)$\;
			
			\If{$\mathcal A_s(t)$ satisfies LICQ}{
				Construct the fixed-active-set KKT system and obtain
				$\mathbf F_{s,t}$,
				$\mathbf h_{s,t}$,
				$\mathbf K_s(t)$,
				$\mathbf k_s(t)$,
				$\bm{\psi}_s(t)$,
				and $\zeta_s(t)$\;
				
			\If{$\mathcal C_{s,t}
				\cap\mathcal U_t^{\mathrm{box}}$
				admits a strictly positive interior margin relative to the hourly
				storage-action space}{
				Add $s$ to $\mathcal R_t^{\mathrm h}$\;
			}
			}
		}
	}
	
	\tcp{Stage 2: Jointly select the hourly regions and ESA actions}
	Construct and solve
	\eqref{eq:intertemporal_region_selection}\;
	
	\KwOut{A globally optimal ESA schedule and the selected hourly critical
		regions.}
\end{algorithm}

\begin{proposition}[Exactness of the intertemporal reformulation]
	\label{prop:on_demand_exploration_complete}
	The intertemporal region-selection problem
	\eqref{eq:intertemporal_region_selection} is equivalent to
	\eqref{best:ESA:nonrestrict} and returns a globally optimal ESA schedule.
\end{proposition}

\begin{IEEEproof}
	We first show that every admissible ESA schedule in the original problem can
	be represented by the hourly region libraries. Consider any
	$\mathbf u\in\mathcal U(\mathbf B)\cap\mathcal U_{\mathrm{mc}}$. For each
	$t\in\mathcal T$, its hourly action satisfies
	$\mathbf u(t)\in\mathcal U_t^{\mathrm{box}}$. By
	Assumption~\ref{assum2}, the market-clearing optimum induced by this action
	has a LICQ-valid active set. Safe screening and exhaustive active-set
	enumeration retain all full-dimensional LICQ-valid hourly critical regions
	whose intersections with $\mathcal U_t^{\mathrm{box}}$ have a strictly
	positive interior margin relative to the hourly storage-action space. If
	$\mathbf u(t)$ lies on a shared boundary, the non-strict boundary convention
	assigns it to at least one adjacent full-dimensional region. Consequently, for
	every admissible schedule and every $t\in\mathcal T$, there exists a region
	in $\mathcal R_t^{\mathrm h}$ containing $\mathbf u(t)$. Thus, every admissible schedule of the original ESA problem admits at least
	one valid hourly-region selection in each period.
	
	We next show that the reformulation introduces no artificial ESA schedules.
	Consider any feasible solution of
	\eqref{eq:intertemporal_region_selection}. Constraint
	\eqref{eq:select_one_hourly_region} selects one region at each time, while
	\eqref{eq:scaled_hourly_action} sets the auxiliary actions of all unselected
	regions to zero. Hence, by~\eqref{eq:hourly_action_disaggregation}, the
	auxiliary action of the selected region equals $\mathbf u(t)$ and satisfies
	the corresponding regional constraints. The constraint
	$\mathbf u\in\mathcal U(\mathbf B)$ enforces the exact intertemporal storage
	constraints, and the objective recovers the corresponding hourly regional ESA
	objective. Thus, every feasible solution of
	\eqref{eq:intertemporal_region_selection} corresponds to a feasible solution
	of~\eqref{best:ESA:nonrestrict} with the same objective value.
	
	Finally, we show that the reformulation excludes no feasible ESA schedule.
	Consider any feasible solution of~\eqref{best:ESA:nonrestrict}. By the coverage
	property established above, at each $t$ there exists a region containing
	$\mathbf u(t)$. Select one such region, assign $\mathbf u(t)$ to its
	corresponding $\mathbf v_s(t)$, and set all other auxiliary actions to zero.
	This produces a feasible solution of
	\eqref{eq:intertemporal_region_selection} with the same objective value.
	Therefore, the two formulations have the same feasible ESA schedules,
	objective values, and globally optimal solutions.
\end{IEEEproof}

\section{Proofs for Section~\ref{parametric}}
\subsection{Proof of Lemma~\ref{lem:critical_region_polyhedron}}
\label{app:proof_critical_region_polyhedron}
\begin{IEEEproof}
Slightly overloading the notation, let \(\mathbf x\) stack
\(\{\mathbf g(t),\mathbf d(t),\mathbf p(t)\}_{t\in\mathcal T}\), and recall
that \(\operatorname{vec}(\mathbf u)\), defined immediately before
Lemma~\ref{lem:critical_region_polyhedron}, vectorizes the ESA schedule over
all buses and periods. The market-clearing problem can be written as the
	following convex parametric quadratic program
\begin{subequations}\label{compact:dcopf}
	\begin{align}
		\min_{\mathbf x}\quad&
		\tfrac12\mathbf x^\top\mathbf Q\mathbf x
		+\mathbf c^\top\mathbf x\\
		\text{s.t.}\quad&
		\bm\lambda:
		\mathbf A\mathbf x
		=
		\mathbf b+\mathbf D\operatorname{vec}(\mathbf u),
		\label{eq:equality_constraint_1}\\
		&
		\bm\nu:
		\mathbf R\mathbf x=\mathbf 0,
		\label{eq:equality_constraint_2}\\
		&
		\bm\mu:
		\mathbf G\mathbf x\leq\mathbf q.
		\label{eq:inequality_constraint}
	\end{align}
\end{subequations}
	Here, equality constraint \eqref{eq:equality_constraint_1} stacks the
	power-balance constraints \eqref{DCOPF:net_power:full}, with
	\(\mathbf D\operatorname{vec}(\mathbf u)\) capturing the affine dependence on
	the ESA schedule. Equality constraint \eqref{eq:equality_constraint_2} stacks the
	lossless system-balance constraints~\eqref{DCOPF:lossless}. The inequalities in
	\eqref{eq:inequality_constraint} stack the generator bounds~\eqref{DCOPF:generator}, demand
	bounds~\eqref{DCOPF:elastic_load}, and transmission line limits
	\eqref{DCOPF:line}.
	The associated dual variables are, respectively, the LMP vector
	\(\bm\lambda\), the multiplier \(\bm\nu\) of the lossless system-balance
	equalities, and the nonnegative multiplier \(\bm\mu\) of the operating-limit
	inequalities.
	
Fix a critical region $\mathcal C_r$, and let $\mathcal A_r$ denote the
market-clearing constraint-binding pattern common to its relative interior,
where $\operatorname{relint}(\mathcal C_r)$ denotes the interior of
$\mathcal C_r$ relative to its affine hull. Thus,
\[
\mathcal A(\mathbf u)=\mathcal A_r,
\qquad
\mathbf u\in\operatorname{relint}(\mathcal C_r).
\]
Let $\mathcal I_r$ denote the complementary inactive rows of $\mathbf G$. The KKT conditions for this pattern are
	\begin{align*}
		\mathbf Q\mathbf x+\mathbf c
		+\mathbf A^\top\bm\lambda
		+\mathbf R^\top\bm\nu
		+\mathbf G_{\mathcal A_r}^\top\bm\mu_{\mathcal A_r}
		&=\mathbf0,\\
		\mathbf A\mathbf x
		&=\mathbf b+\mathbf D\operatorname{vec}(\mathbf u),\\
		\mathbf R\mathbf x&=\mathbf0,\\
		\mathbf G_{\mathcal A_r}\mathbf x
		&=\mathbf q_{\mathcal A_r},\\
		\mathbf G_{\mathcal I_r}\mathbf x
		&\leq\mathbf q_{\mathcal I_r},\\
		\bm\mu_{\mathcal A_r}&\geq\mathbf0.
	\end{align*}
	These conditions are linear equalities and inequalities in
	\((\operatorname{vec}(\mathbf u),\mathbf x,\bm\lambda,\bm\nu,
	\bm\mu_{\mathcal A_r})\). The feasible set is
	polyhedral, so the first-order optimality condition and the normal-cone
	formula for polyhedra yield multipliers satisfying the KKT system at every
	optimizer; conversely, convexity makes the KKT conditions sufficient
	\cite[Secs.~23 and 28]{rockafellar1970convex}.  Thus, the displayed
	conditions are necessary and sufficient for optimality under the fixed
	binding pattern associated with region $r$.  
	
Consequently, the set of all tuples satisfying these conditions is a
polyhedron. Its projection onto the coordinates corresponding to
\(\operatorname{vec}(\mathbf u)\) is also a polyhedron
\cite[Thm.~19.3]{rockafellar1970convex}. This projection is the closed
critical region generated by the binding pattern $\mathcal A_r$ on its
relative interior; under the non-strict boundary convention used in the
paper, it is $\mathcal C_r$. Every polyhedron admits a finite half-space
representation \cite[Sec.~19]{rockafellar1970convex}; hence there exist
\(\mathbf F_r\) and \(\mathbf h_r\) such that
\[
\mathcal C_r
=
\left\{
\mathbf u:
\mathbf F_r\operatorname{vec}(\mathbf u)
\leq\mathbf h_r
\right\}.
\]
	The non-strict inequalities include ESA schedules on boundaries shared by
	adjacent critical regions.
\end{IEEEproof}
\subsection{Proof of Theorem~\ref{thm:affine_mapping}}
\label{app:proof_affine_mapping}

\begin{IEEEproof}
	We prove the three claims in order. Fix \(t\in\mathcal T\) and a critical
	region \(\mathcal C_r\). All quantities below correspond to this fixed
	pair \((t,r)\), so time arguments are suppressed where no ambiguity arises.
	
	\emph{Part 1: Regional affine representation.}
	Let \(\mathcal G_r^0\), \(\mathcal G_r^+\), and \(\mathcal G_r^-\) denote
	the generator indices that are, respectively, strictly interior, at their
	upper bounds, and at their lower bounds throughout the relative interior of
	\(\mathcal C_r\). Define \(\mathcal D_r^0\), \(\mathcal D_r^+\), and
	\(\mathcal D_r^-\) analogously for demand. When coincident lower
	and upper bounds represent a fixed quantity, we treat that quantity as an
	equality, equivalently deleting one of the two redundant bound inequalities,
	before applying LICQ.
	
	Stationarity with respect to generation and demand gives 
	\begin{equation}
		\mathbf g^\star
		=\rev{\mathbf E_r^{\mathrm g}}\bm\lambda^\star+\rev{\bm\eta_r^{\mathrm g}},
		\qquad
		\mathbf d^\star
		=-\rev{\mathbf E_r^{\mathrm d}}\bm\lambda^\star+\rev{\bm\eta_r^{\mathrm d}},
		\label{eq:appendix_gd_affine}
	\end{equation}
where these matrices are diagonal and, together with their affine terms, are
defined componentwise by
	\begin{align}
		[\rev{\mathbf E_r^{\mathrm g}}]_{ii}
		&=
		\begin{cases}
			1/\alpha_i^{\mathrm g},&i\in\mathcal G_r^0,\\
			0,&i\in\mathcal G_r^+\cup\mathcal G_r^-,
		\end{cases}
		\label{eq:appendix_Dg}\\
		[\rev{\bm\eta_r^{\mathrm g}}]_i
		&=
		\begin{cases}
			-\beta_i^{\mathrm g}/\alpha_i^{\mathrm g},&i\in\mathcal G_r^0,\\
			\overline g_i,&i\in\mathcal G_r^+,\\
			\underline g_i,&i\in\mathcal G_r^-,
		\end{cases}
		\label{eq:appendix_cg}\\
		[\rev{\mathbf E_r^{\mathrm d}}]_{ii}
		&=
		\begin{cases}
			1/\alpha_i^{\mathrm d},&i\in\mathcal D_r^0,\\
			0,&i\in\mathcal D_r^+\cup\mathcal D_r^-,
		\end{cases}
		\label{eq:appendix_Dd}\\
		[\rev{\bm\eta_r^{\mathrm d}}]_i
		&=
		\begin{cases}
			\beta_i^{\mathrm d}/\alpha_i^{\mathrm d},&i\in\mathcal D_r^0,\\
			\overline d_i,&i\in\mathcal D_r^+,\\
			\underline d_i,&i\in\mathcal D_r^-.
		\end{cases}
		\label{eq:appendix_cd}
	\end{align}
	If a lower and upper bound coincide, the fixed coordinate may be assigned
	to either bound set, without changing these formulas. The expressions in
	\eqref{eq:appendix_Dg}--\eqref{eq:appendix_cd} follow from the
	interior stationarity equations and from fixing every bound-active variable
	at the corresponding limit.
	
	Define
	\begin{equation}
		\rev{\mathbf E_r^{\mathrm e}}:=
		\rev{\mathbf E_r^{\mathrm g}}+\rev{\mathbf E_r^{\mathrm d}},
		\qquad
		\rev{\bm\eta_r^{\mathrm e}}:=
		\rev{\bm\eta_r^{\mathrm g}}-\rev{\bm\eta_r^{\mathrm d}}.\nonumber
	\end{equation}
	The nodal power-balance equations imply
	\begin{equation}
		\mathbf p^\star
		=\rev{\mathbf E_r^{\mathrm e}}\bm\lambda^\star
		+\rev{\bm\eta_r^{\mathrm e}}-\mathbf u.
		\label{eq:appendix_net_injection_affine}
	\end{equation}
	
	Define the active upper- and lower-limit index sets by
	\[
	\mathcal L_r^+
	:=\{\ell:\mathbf H_{\ell:}\mathbf p^\star=\overline f_\ell\},
	\qquad
	\mathcal L_r^-
	:=\{\ell:\mathbf H_{\ell:}\mathbf p^\star=-\overline f_\ell\}. 
	\]
	\rev{Here \(\mathbf H_{\ell:}\) denotes the \(\ell\)th row of \(\mathbf H\).}
	\rev{Let
		\(\mathbf P:=\mathbf I-|\mathcal N|^{-1}\mathbf1\mathbf1^\top\) and
		\(\widehat{\mathbf H}:=\mathbf H\mathbf P\). Since
		\(\mathbf1^\top\mathbf p^\star=0\), replacing \(\mathbf H\) by
		\(\widehat{\mathbf H}\) leaves line flows unchanged on feasible injections,
		while making the subsequent formula compatible with eliminating the
		system-balance degree of freedom. Moreover, each row of
		\(\widehat{\mathbf H}\) differs from the corresponding row of \(\mathbf H\)
		by a multiple of \(\mathbf1^\top\), so this replacement only shifts the
		system-balance multiplier in the stationarity equations.}	
	\rev{Construct the active-line matrix by stacking the rows indexed by
		\(\mathcal L_r^+\) and the negated rows indexed by \(\mathcal L_r^-\):}
	\begin{equation}
		\rev{\mathbf W_r}:=
		\begin{bmatrix}
			\widehat{\mathbf H}_{\mathcal L_r^+:}\\
			-\widehat{\mathbf H}_{\mathcal L_r^-:}
		\end{bmatrix},
		\qquad
		\rev{\bm\omega_r}:=
		\begin{bmatrix}
			\overline{\mathbf f}_{\mathcal L_r^+}\\
			\overline{\mathbf f}_{\mathcal L_r^-}
		\end{bmatrix}.\nonumber
	\end{equation}

	Thus the active line limits are precisely
	\(\rev{\mathbf W_r}\mathbf p^\star=\rev{\bm\omega_r}\). Let \(\bm\mu_r^\star\)
	stack their multipliers in the same row order. Stationarity with respect to
	\(\mathbf p\) yields
	\begin{equation}
		\bm\lambda^\star
		=-\nu^\star\mathbf1-\rev{\mathbf W_r}^\top\bm\mu_r^\star.
		\label{eq:appendix_lambda_duals}
	\end{equation}
	
	Define
	\begin{align}
		\rev{\mathbf N_r}:=
		\begin{bmatrix}
			\mathbf1^\top\\
			\rev{\mathbf W_r}
		\end{bmatrix},
		\qquad
		\rev{\bm\rho_r}:=
		\begin{bmatrix}
			0\\
			\rev{\bm\omega_r}
		\end{bmatrix},
		\nonumber\\
		\mathbf y_r:=
		\begin{bmatrix}
			\nu^\star\\
			\bm\mu_r^\star
		\end{bmatrix},
		\qquad
		\mathbf M_r
		:=\rev{\mathbf N_r}\rev{\mathbf E_r^{\mathrm e}}\rev{\mathbf N_r}^\top.\nonumber
	\end{align}
	Substituting \eqref{eq:appendix_lambda_duals} into
	\eqref{eq:appendix_net_injection_affine} and imposing
	\(\rev{\mathbf N_r}\mathbf p^\star=\rev{\bm\rho_r}\) gives
	\begin{equation}
		\mathbf M_r\mathbf y_r
		=\rev{\bm\chi_r}-\rev{\mathbf N_r}\mathbf u.
		\label{eq:appendix_M_system}
	\end{equation}
	The vector in \eqref{eq:appendix_M_system} is
	\begin{equation}
		\rev{\bm\chi_r}
		:=\rev{\mathbf N_r}\rev{\bm\eta_r^{\mathrm e}}-\rev{\bm\rho_r}
		=
		\begin{bmatrix}
			\mathbf1^\top\rev{\bm\eta_r^{\mathrm e}}\\
			\rev{\mathbf W_r}\rev{\bm\eta_r^{\mathrm e}}-\rev{\bm\omega_r}
		\end{bmatrix}.\nonumber
	\end{equation}
	
	\rev{\emph{Claim.}
		The restricted matrix \(\mathbf N_{r,\mathcal J_r}\) has full row rank, where
		\(\mathcal J_r:=\{i:[\mathbf E_r^{\mathrm e}]_{ii}>0\}\) and
		\(\mathbf N_{r,\mathcal J_r}\) denotes the restriction of \(\mathbf N_r\) to
		columns indexed by \(\mathcal J_r\).}
	
	\rev{\emph{Proof of claim.}}
	Suppose \(\rev{\mathbf N_{r,\mathcal J_r}}\) does not have full row rank. Then
	there exists a nonzero vector \(\mathbf z\) such that
	\(\rev{\mathbf N_{r,\mathcal J_r}}^\top\mathbf z=\mathbf0\). Set
	\(\mathbf w:=\rev{\mathbf N_r}^\top\mathbf z\), so \(w_i=0\) for every
	\(i\in\rev{\mathcal J_r}\). \rev{Since the rows of \(\mathbf W_r\) differ
		from the corresponding signed active rows of \(\mathbf H\) only by multiples
		of the system-balance row, this is equivalent to a linear combination of the
		original system-balance and active-line gradients.} Take the linear combination specified by
	\(\mathbf z\) of the system-balance and active-line gradients, whose
	\(\mathbf p\)-block is \(\mathbf w\), and add the nodal-balance gradients
	with coefficient vector \(-\mathbf w\). This cancels the
	\(\mathbf p\)-block and leaves nonzero entries only in the generation and
	demand coordinates indexed by \(i\notin\rev{\mathcal J_r}\). At each such bus,
	both generation and demand are fixed at active bounds, because
	\([\rev{\mathbf E_r^{\mathrm e}}]_{ii}=0\). Their active-bound gradients can
	therefore cancel the remaining entries. The resulting coefficients are
	not all zero because \(\mathbf z\ne\mathbf0\), so this is a nontrivial
	linear dependence among the active-constraint gradients, contradicting
	Assumption~\ref{assum2}. 
	\rev{This proves the claim.} Hence
	\(\rev{\mathbf N_{r,\mathcal J_r}}\) has full row rank. Assumption~\ref{assum1}
	implies that the positive principal submatrix
	\(\rev{\mathbf E_{r,\mathcal J_r}^{\mathrm e}}\) is positive definite, and thus
	\begin{equation}
		\mathbf M_r
		=
		\rev{\mathbf N_{r,\mathcal J_r}}
		\rev{\mathbf E_{r,\mathcal J_r}^{\mathrm e}}
		\rev{\mathbf N_{r,\mathcal J_r}}^\top
		\succ\mathbf0.
		\label{eq:appendix_M_positive_definite}
	\end{equation}
	
	Solving \eqref{eq:appendix_M_system} and substituting into
	\eqref{eq:appendix_lambda_duals} gives
	\[
	\bm\lambda^\star
	=\rev{\mathbf N_r}^\top\mathbf M_r^{-1}\rev{\mathbf N_r}\mathbf u
	-\rev{\mathbf N_r}^\top\mathbf M_r^{-1}\rev{\bm\chi_r}.
	\]
	Accordingly, define
	\begin{equation}
		\mathbf K_r(t):=\rev{\mathbf N_r}^\top\mathbf M_r^{-1}\rev{\mathbf N_r},
		\qquad
		\mathbf k_r(t):=-\rev{\mathbf N_r}^\top\mathbf M_r^{-1}\rev{\bm\chi_r}.
		\label{eq:appendix_K_k}
	\end{equation}
	Then
	\[
	\bm\lambda^\star(t;\mathbf u)
	=\mathbf K_r(t)\mathbf u(t)+\mathbf k_r(t).
	\]
	
	Substitution into \eqref{eq:appendix_gd_affine} and
	\eqref{eq:appendix_net_injection_affine} gives
	\[
	\mathbf x^\star(t;\mathbf u)
	=\bm\Gamma_r(t)\mathbf u(t)+\bm\gamma_r(t).
	\]
	For
	\(\mathbf x^\star=[(\mathbf g^\star)^\top,
	(\mathbf d^\star)^\top,(\mathbf p^\star)^\top]^\top\), the exact
	expressions are
	\begin{equation}
		\bm\Gamma_r(t)=
		\begin{bmatrix}
			\rev{\mathbf E_r^{\mathrm g}}\mathbf K_r(t)\\
			-\rev{\mathbf E_r^{\mathrm d}}\mathbf K_r(t)\\
			\rev{\mathbf E_r^{\mathrm e}}\mathbf K_r(t)-\mathbf I
		\end{bmatrix},
		\qquad
		\bm\gamma_r(t)=
		\begin{bmatrix}
			\rev{\mathbf E_r^{\mathrm g}}\mathbf k_r(t)+\rev{\bm\eta_r^{\mathrm g}}\\
			-\rev{\mathbf E_r^{\mathrm d}}\mathbf k_r(t)+\rev{\bm\eta_r^{\mathrm d}}\\
			\rev{\mathbf E_r^{\mathrm e}}\mathbf k_r(t)+\rev{\bm\eta_r^{\mathrm e}}
		\end{bmatrix}.\nonumber
	\end{equation}
	The derivation above applies on the relative interior of
	\(\mathcal C_r\). Since the regional formulas are affine and the
	\rev{market-clearing primal-dual solution is unique by the strict convexity in
		Assumption~\ref{assum1} together with the LICQ regularity in
		Assumption~\ref{assum2}}, taking limits from the relative
	interior extends both formulas to every boundary point of
	\(\mathcal C_r\).
	This proves part 1).

	\emph{Part 2: Continuity and piecewise-affine structure.}
	Only finitely many binding patterns are possible, and Part~1 gives an affine
	formula on every associated nonempty critical region. Under
	Assumptions~\ref{assum1}--\ref{assum2}, the lower-level primal solution and
	LMP vector are unique. Hence, at a point shared by adjacent critical
	regions, both regional formulas equal the same primal-dual solution and
	therefore agree. By definition of \(\mathcal U_{\mathrm{mc}}\), the lower-level problem
	is feasible for every \(\mathbf u\in\mathcal U_{\mathrm{mc}}\); Assumption~\ref{assum3}
	restricts the ESA's submitted schedules to this clearable set. Consequently,
	\(\mathbf x^\star(t;\mathbf u)\) and
	\(\bm\lambda^\star(t;\mathbf u)\) are continuous piecewise-affine functions
	on \(\mathcal U_{\mathrm{mc}}\), proving part 2).
	
	\emph{Part 3: Symmetry and positive semidefiniteness.}
	\rev{By \eqref{eq:appendix_M_positive_definite}, the reduced matrix
		\(\mathbf M_r=
		\mathbf N_{r,\mathcal J_r}\mathbf E_{r,\mathcal J_r}^{\mathrm e}
		\mathbf N_{r,\mathcal J_r}^{\top}\) is symmetric positive definite. Hence
		\(\mathbf M_r^{-1}\) is also symmetric positive definite. Using
		\eqref{eq:appendix_K_k},}
	\[
	\rev{\mathbf K_r(t)^\top
		=\left(\mathbf N_r^\top\mathbf M_r^{-1}\mathbf N_r\right)^\top
		=\mathbf N_r^\top\mathbf M_r^{-1}\mathbf N_r
		=\mathbf K_r(t).}
	\]

	\rev{Moreover,} for every
	\(\mathbf z\in\mathbb R^{|\mathcal N|}\),
	\[
	\mathbf z^\top\mathbf K_r(t)\mathbf z
	=(\rev{\mathbf N_r}\mathbf z)^\top\mathbf M_r^{-1}
	(\rev{\mathbf N_r}\mathbf z)\geq0,
	\]
	because \(\rev{\mathbf M_r^{-1}\succ\mathbf0}\). Thus
	\(\mathbf K_r(t)\succeq\mathbf0\), proving part 3). 
	
Since \(t\) was arbitrary, the results of all three parts hold for every
\(t\in\mathcal T\), completing the proof of
Theorem~\ref{thm:affine_mapping}.
\end{IEEEproof}

\section{Proofs for Section~\ref{sec:esa_binding_pattern_decomposition}}

\subsection{Proof of Lemma~\ref{lem:regional_esa_objective} and
	Corollary~\ref{cor:esa_objective_continuity}}
\label{app:proof_regional_objective}
\begin{IEEEproof}[Proof of Lemma~\ref{lem:regional_esa_objective}]
	Fix \(r\in\mathcal R\) and \(\mathbf u\in\mathcal C_r\).  For each
	\(t\in\mathcal T\), substitute
	\(\bm\lambda^\star(t;\mathbf u)
	=\mathbf K_r(t)\mathbf u(t)+\mathbf k_r(t)\)
	from~\eqref{eq:lambda_affine} into~\eqref{eq:esa_objective_global}.  Using the
	current source-to-sink FTR convention,
	\[
	\lambda_j^\star(t;\mathbf u)-\lambda_i^\star(t;\mathbf u)
	=
	(\mathbf e_j-\mathbf e_i)^\top
	\bigl(\mathbf K_r(t)\mathbf u(t)+\mathbf k_r(t)\bigr).
	\]
The contribution of period \(t\) to the ESA objective is therefore
\begin{align*}
	&\mathbf u(t)^\top\mathbf K_r(t)\mathbf u(t)
	+\mathbf k_r(t)^\top\mathbf u(t)\\
	&\quad-
	\sum_{(i,j)\in\mathcal F}
	\kappa_{(i,j)}(t)
	(\mathbf e_j-\mathbf e_i)^\top
	\bigl(\mathbf K_r(t)\mathbf u(t)+\mathbf k_r(t)\bigr)\\
	&=
	\mathbf u(t)^\top\mathbf K_r(t)\mathbf u(t)
	+\bm\psi_r(t)^\top\mathbf u(t)+\zeta_r(t),
\end{align*}
with \(\bm\psi_r(t)\) and \(\zeta_r(t)\) given by
\eqref{eq:regional_linear_term} and~\eqref{eq:regional_constant_term}.
Summing over \(t\) proves~\eqref{eq:regional_esa_objective_quadratic}.
Theorem~\ref{thm:affine_mapping} gives
\(\mathbf K_r(t)\succeq\mathbf0\) for every \(t\); hence \(J_r\) is convex.
\end{IEEEproof}

\begin{IEEEproof}[Proof of Corollary~\ref{cor:esa_objective_continuity}]
	Lemma~\ref{lem:regional_esa_objective} establishes
	\(J(\mathbf u)=J_r(\mathbf u)\) on each \(\mathcal C_r\). \textcolor{black}{If \(\mathbf u\in\mathcal C_r\cap\mathcal C_s\), then by part~2 of
		Theorem~\ref{thm:affine_mapping}, the regional affine LMP representations
		agree at \(\mathbf u\). Hence,
		\[
		\mathbf K_r(t)\mathbf u(t)+\mathbf k_r(t)
		=
		\mathbf K_s(t)\mathbf u(t)+\mathbf k_s(t),
		\qquad t\in\mathcal T.
		\]
	}
	Moreover, note that both regional ESA formulas are obtained by evaluating the same expression
	\eqref{eq:esa_objective_global} at this common schedule $\mathbf{u}$ and LMP trajectory.
	Thus \(J_r(\mathbf u)=J_s(\mathbf u)\) on every common boundary.  Since the
	finite collection \(\{\mathcal C_r\}_{r\in\mathcal R}\) covers
	\(\mathcal U_{\mathrm{mc}}\), and each \(J_r\) is quadratic and hence continuous, the
	boundary agreement proves that \(J\) is continuous and piecewise quadratic
	on \(\mathcal U_{\mathrm{mc}}\).
\end{IEEEproof}

\subsection{Proof of Lemma~\ref{lem:regional_price_uniqueness} and
	Theorem~\ref{thm:critical_region_decomposition}}
\label{app:proof_decomposition}
\begin{IEEEproof}[Proof of Lemma~\ref{lem:regional_price_uniqueness}]
Fix \(r\) and \(\mathbf B\), and let
\(\mathbf u^1,\mathbf u^2\in\mathcal U_r^\star(\mathbf B)\). The regional
	feasible set is convex, so
	\(\overline{\mathbf u}:=(\mathbf u^1+\mathbf u^2)/2\) is feasible.  Convexity
	of \(J_r\) and optimality of \(\mathbf u^1,\mathbf u^2\) imply that
	\(\overline{\mathbf u}\) is also optimal.  Let
	\(\rev{\Delta\mathbf u(t):=\mathbf u^1(t)-\mathbf u^2(t)}\). Direct expansion of the
	quadratic representation~\eqref{eq:regional_esa_objective_quadratic} gives
	\[
	\frac{J_r(\mathbf u^1)+J_r(\mathbf u^2)}{2}
	-J_r(\overline{\mathbf u})
	=
	\frac14\sum_{t\in\mathcal T}
	\rev{\Delta\mathbf u(t)}^\top\mathbf K_r(t)\rev{\Delta\mathbf u(t)}.
	\]
	The left-hand side is zero because all three schedules are optimal.  Every
	term on the right is nonnegative by
	\(\mathbf K_r(t)\succeq\mathbf0\), and therefore
	\[
	\rev{\Delta\mathbf u(t)}^\top\mathbf K_r(t)\rev{\Delta\mathbf u(t)}=0,
	\qquad t\in\mathcal T.
	\]
	For a symmetric positive-semidefinite matrix, this equality implies
	\(\mathbf K_r(t)\rev{\Delta\mathbf u(t)}=\mathbf0\).  Consequently,
	\[
	\mathbf K_r(t)\mathbf u^1(t)+\mathbf k_r(t)
	=
	\mathbf K_r(t)\mathbf u^2(t)+\mathbf k_r(t),
	\qquad t\in\mathcal T.
	\]
	By~\eqref{eq:lambda_affine}, the two schedules induce the same LMP trajectory.
Since the pair of optimizers was arbitrary, the conclusion holds for all
members of \(\mathcal U_r^\star(\mathbf B)\).
\end{IEEEproof}

\begin{IEEEproof}[Proof of Theorem~\ref{thm:critical_region_decomposition}]
	The critical regions cover the ISO-clearable parameter space:
	\[
	\mathcal U_{\mathrm{mc}}
	=
	\bigcup_{r\in\mathcal R}\mathcal C_r.
	\]
	Under Assumption~\ref{assum3}, the feasible ESA schedules in the Stackelberg
	problem are therefore
	\[
	\mathcal U(\mathbf B)\cap\mathcal U_{\mathrm{mc}}
	=
	\bigcup_{r\in\mathcal R}
	\left(
	\mathcal U(\mathbf B)\cap\mathcal C_r
	\right).
	\]
By Lemma~\ref{lem:critical_region_polyhedron}, the \(r\)-th set in this union
is described by
\(\mathbf u\in\mathcal U(\mathbf B)\) and
\(\mathbf F_r\operatorname{vec}(\mathbf u)\leq\mathbf h_r\). By
	Corollary~\ref{cor:esa_objective_continuity}, the global ESA objective equals
	\(J_r\) on that set.  It follows that
\begin{align*}
	\min_{\mathbf u\in\mathcal U(\mathbf B)\cap\mathcal U_{\mathrm{mc}}}J(\mathbf u)
	&=
	\min_{r\in\mathcal R}
	\;
	\min_{\substack{
			\mathbf u\in\mathcal U(\mathbf B)\\
			\mathbf F_r\operatorname{vec}(\mathbf u)\leq\mathbf h_r}}
	J_r(\mathbf u),
\end{align*}
	where an infeasible regional problem has value \(+\infty\).  This is exactly
	\eqref{best:ESA:nonrestrict}.  Hence any optimizer from a region attaining the
	outer minimum is globally optimal for the ESA.  \rev{Whenever a regional feasible
		set is nonempty, the corresponding minimum is attained because
		\(\mathcal U(\mathbf B)\cap\mathcal C_r\) is closed and bounded and
		\(J_r\) is continuous.} Together with the unique ISO market-clearing response, \rev{whose uniqueness
		follows from} Assumptions~\ref{assum1}--\ref{assum2}, such a schedule forms an
		equilibrium of the ESA--ISO Stackelberg interaction.
\end{IEEEproof}
\section{Proofs for Section~\ref{sec:pattern_manipulation_ftr}}
\subsection{Proof of Proposition~\ref{prop:price_taking_implementation}}

\begin{IEEEproof}
	Write the centralized-storage benchmark
	\eqref{eq:centralized_storage_benchmark} equivalently as the minimization of
	generation cost minus demand utility. For each bus $i\in\mathcal N$, let
	$\underline{\bm\sigma}_i^{\mathrm{cen}}\geq\mathbf 0$ and
	$\overline{\bm\sigma}_i^{\mathrm{cen}}\geq\mathbf 0$ denote the Lagrange
	multipliers associated with
	\[
	-\mathbf L\mathbf u_i\leq\mathbf 0,
	\qquad
	\mathbf L\mathbf u_i-B_i\mathbf 1\leq\mathbf 0,
	\]
	respectively, and let $\tau_i^{\mathrm{cen}}$ denote the multiplier associated
	with
	$
	\mathbf 1^\top\mathbf u_i=0.
	$
	
	At an optimal solution of
	\eqref{eq:centralized_storage_benchmark}, the KKT stationarity condition with
	respect to $\mathbf u_i$ is
	\begin{equation}
		\bm\lambda_i^{\mathrm{cen}}
		-\mathbf L^\top\underline{\bm\sigma}_i^{\mathrm{cen}}
		+\mathbf L^\top\overline{\bm\sigma}_i^{\mathrm{cen}}
		+\tau_i^{\mathrm{cen}}\mathbf 1
		=
		\mathbf 0,
		\qquad
		\forall i\in\mathcal N,
		\label{eq:centralized_u_stationarity}
	\end{equation}
	where $\bm\lambda_i^{\mathrm{cen}}$ is the LMP trajectory at bus $i$.
	The remaining KKT conditions associated with the storage variables are
	\begin{align}
		&
		\mathbf 0
		\leq
		\mathbf L\mathbf u_i^{\mathrm{cen}}
		\leq
		B_i\mathbf 1,
		\qquad
		\mathbf 1^\top\mathbf u_i^{\mathrm{cen}}
		=
		0,
		\label{eq:centralized_u_primal}
		\\
		&
		\underline{\bm\sigma}_i^{\mathrm{cen}}
		\geq
		\mathbf 0,
		\qquad
		\overline{\bm\sigma}_i^{\mathrm{cen}}
		\geq
		\mathbf 0,
		\label{eq:centralized_u_dual}
		\\
		&
		\left(\underline{\bm\sigma}_i^{\mathrm{cen}}\right)^\top
		\mathbf L\mathbf u_i^{\mathrm{cen}}
		=
		0,
		\qquad
		\left(\overline{\bm\sigma}_i^{\mathrm{cen}}\right)^\top
		\left(
		\mathbf L\mathbf u_i^{\mathrm{cen}}
		-B_i\mathbf 1
		\right)
		=
		0.
		\label{eq:centralized_u_cs}
	\end{align}
	
	Now consider a price-taking ESA facing the fixed LMP trajectory
	$\bm\lambda^{\mathrm{cen}}$:
	\begin{equation}
		\min_{\mathbf u\in\mathcal U(\mathbf B)}
		\sum_{t\in\mathcal T}
		\bm\lambda^{\mathrm{cen}}(t)^\top\mathbf u(t).
		\label{eq:price_taking_auxiliary}
	\end{equation}
	Its stationarity condition with respect to $\mathbf u_i$ is exactly
	\eqref{eq:centralized_u_stationarity}. Moreover, because
	\eqref{eq:price_taking_auxiliary} has the same storage constraints
	$\mathbf u\in\mathcal U(\mathbf B)$, its primal-feasibility,
	dual-feasibility, and complementary-slackness conditions are exactly
	\eqref{eq:centralized_u_primal}--\eqref{eq:centralized_u_cs}.
	Hence,
	$\mathbf u^{\mathrm{cen}}$, together with
	$\underline{\bm\sigma}^{\mathrm{cen}}$,
	$\overline{\bm\sigma}^{\mathrm{cen}}$, and
	$\bm\tau^{\mathrm{cen}}$, satisfies all KKT conditions of
	\eqref{eq:price_taking_auxiliary}. Since this problem is convex, the KKT
	conditions are sufficient for optimality. Therefore,
	\begin{equation}
		\mathbf u^{\mathrm{cen}}
		\in
		\operatorname*{arg\,min}_{\mathbf u\in\mathcal U(\mathbf B)}
		\sum_{t\in\mathcal T}
		\bm\lambda^{\mathrm{cen}}(t)^\top\mathbf u(t).
		\label{eq:price_taking_larger_set}
	\end{equation}
	
	Equation~\eqref{eq:price_taking_larger_set} establishes optimality over
	$\mathcal U(\mathbf B)$, whereas
	Proposition~\ref{prop:price_taking_implementation} additionally requires the
	schedule to be ISO-clearable, resulting in the feasible set
	$\mathcal U(\mathbf B)\cap\mathcal U_{\mathrm{mc}}$. Since
	$\mathbf u^{\mathrm{cen}}$ is part of a feasible solution of
	\eqref{eq:centralized_storage_benchmark}, it belongs to $\mathcal U_{\mathrm{mc}}$ and
	hence to $\mathcal U(\mathbf B)\cap\mathcal U_{\mathrm{mc}}$. Moreover, because
	$\mathcal U(\mathbf B)\cap\mathcal U_{\mathrm{mc}}\subseteq\mathcal U(\mathbf B)$,
	optimality over $\mathcal U(\mathbf B)$ implies optimality over this
	intersection.
	Thus,
	\[
	\mathbf u^{\mathrm{cen}}
	\in
	\operatorname*{arg\,min}_{\mathbf u\in
		\mathcal U(\mathbf B)\cap\mathcal U_{\mathrm{mc}}}
	\sum_{t\in\mathcal T}
	\bm\lambda^{\mathrm{cen}}(t)^\top\mathbf u(t).
	\]
	
	Finally, for fixed $\bm\lambda^{\mathrm{cen}}$, the FTR term in
	\eqref{eq:price_taking_implementation} is independent of $\mathbf u$ and
	therefore does not affect the optimizer. Hence,
	\eqref{eq:price_taking_implementation} follows.
\end{IEEEproof}

\subsection{Proof of Theorem~\ref{thm:ftr_welfare_implications}}
\label{app:proof_ftr_welfare_implications}
\begingroup
\begin{IEEEproof}
	We prove the two statements separately.
	
	\emph{Part 1: $\bm{\kappa}=\mathbf 0$.}
	For a given submitted ESA schedule \(\mathbf u\), define the optimal-value
	function of the ISO market-clearing problem as
	\begin{subequations}
		\label{eq:iso_value_function}
		\begin{align}
			\Phi(\mathbf u)
			:=
			\min_{\mathbf g,\mathbf d,\mathbf p}
			\quad
			&
			\sum_{t\in\mathcal T}
			\left[
			C_t\!\left(\mathbf g(t)\right)
			-
			V_t\!\left(\mathbf d(t)\right)
			\right]
			\label{eq:iso_value_function_obj}
			\\
			\text{s.t.}\quad
			&
			\eqref{DCOPF:net_power:full},
			~
			\eqref{DCOPF:Constraints}.
			\label{eq:iso_value_function_con}
		\end{align}
	\end{subequations}
	By the definition of social welfare in~\eqref{eq:social_welfare},
	\(\Phi(\mathbf u)=-\mathrm{SW}(\mathbf u)\). The function
	$\Phi$ is convex in $\mathbf u$ by the standard value-function result that
	partial minimization of a jointly convex function over affine constraints
	preserves convexity~\cite[Sec.~3.2.5]{boyd2004convex}: the objective in
	\eqref{eq:iso_value_function_obj} is convex in the ISO variables, and
	$\mathbf u$ enters only affinely in the nodal-balance constraints. Moreover,
	the sensitivity theorem for convex programs gives the Lagrange multiplier of
	\eqref{DCOPF:net_power:full} as a subgradient of the optimal-value function
	with respect to the right-hand-side perturbation~\cite[Sec.~5.6.3]{boyd2004convex}.
	With the sign convention in \eqref{DCOPF:net_power:full}, this gives
	\begin{equation}
		\nabla_{\mathbf u(t)}\Phi(\mathbf u)
		=
		\bm{\lambda}^{\star}(t;\mathbf u),
		\qquad t\in\mathcal T,\nonumber
	\end{equation}
	whenever the value function is differentiable. Under
	Assumption~\ref{assum2}, LICQ implies uniqueness of the KKT multipliers, and
	hence the LMP vector used in the sensitivity formula is well defined. Even
	without invoking differentiability, the same sensitivity theorem gives
	\(\{\bm\lambda^\star(t;\mathbf u)\}_{t\in\mathcal T}\) as a subgradient of
	\(\Phi\), which is sufficient for the supporting-hyperplane inequality used
	below.\textcolor{black}{\footnote{This argument is analogous to the convex under-estimator
			argument used in Proposition~2 of the proof of Theorem~1
			in~\cite{contreras2017participation}.}}
	
	The supporting-hyperplane inequality for $\Phi$, evaluated at an arbitrary
	feasible $\mathbf u$ and compared with $\mathbf 0$, therefore yields
	\begin{align}
		\Phi(\mathbf 0)
		&\geq
		\Phi(\mathbf u)
		+
		\sum_{t\in\mathcal T}
		\bm{\lambda}^{\star}(t;\mathbf u)^\top
		\left(\mathbf 0-\mathbf u(t)\right),
		\nonumber\\
		\mathrm{SW}(\mathbf u)
		&\geq
		\mathrm{SW}(\mathbf 0)
		-
		\sum_{t\in\mathcal T}
		\bm{\lambda}^{\star}(t;\mathbf u)^\top\mathbf u(t).
		\label{eq:welfare_payment_bound}
	\end{align}

	When $\bm{\kappa}=\mathbf 0$, the ESA objective in
	\eqref{bilevel:obj} is precisely its net energy-market payment. Since
	$\mathbf 0\in\mathcal U(\mathbf B)$ and the zero schedule is ISO-clearable,
	optimality of any equilibrium schedule $\mathbf u^\star$ gives
	\begin{equation}
		\sum_{t\in\mathcal T}
		\bm{\lambda}^{\star}(t;\mathbf u^\star)^\top
		\mathbf u^\star(t)
		\leq
		\sum_{t\in\mathcal T}
		\bm{\lambda}^{\star}(t;\mathbf 0)^\top\mathbf 0
		=
		0.
		\label{eq:zero_schedule_outside_option}
	\end{equation}
	Applying \eqref{eq:welfare_payment_bound} to $\mathbf u^\star$ and then using
	\eqref{eq:zero_schedule_outside_option} gives
	\[
	\mathrm{SW}(\mathbf u^\star)
	\geq
	\mathrm{SW}(\mathbf 0)
	=
	\mathrm{SW}^{\mathrm{cen}}(\mathbf 0).
	\]
	The equality holds because fixing $\mathbf u=\mathbf 0$ reproduces exactly
	the no-storage market-clearing problem. On the other hand, the ISO dispatch
	induced by $\mathbf u^\star$ is feasible for the centralized problem
	\eqref{eq:centralized_storage_benchmark}. Since that problem maximizes welfare
	over every $\mathbf u\in\mathcal U(\mathbf B)$ and every associated feasible
	market dispatch,
	\[
	\mathrm{SW}^{\mathrm{cen}}(\mathbf B)
	\geq
	\mathrm{SW}(\mathbf u^\star).
	\]
	Combining the preceding inequalities proves \eqref{eq:welfare_without_ftr}.
	
	\emph{Part 2: a counterexample with a nonzero FTR position.}
	Consider a two-bus, two-period system with one lossless line directed from bus
	1 to bus 2 and having capacity $\overline f=4$. In each period, bus 2 has a
	fixed demand of $3$, bus 1 has no demand, and the utility of the fixed demand
	is normalized to zero. The generator costs are
	\begin{align*}
		C_{1,1}(g)&=\tfrac12 g^2,
		&
		C_{2,1}(g)&=4g+\tfrac92 g^2,\\
		C_{1,2}(g)&=g+\tfrac12 g^2,
		&
		C_{2,2}(g)&=4g+\tfrac92 g^2,
	\end{align*}
	with nonnegative outputs and sufficiently large generation capacities. The
	ESA controls one storage unit at bus 2 with energy capacity $B_2=2$ and
	$B_1=0$. Its feasible schedules are therefore
	\[
	u_2(1)=x,\qquad u_2(2)=-x,\qquad 0\leq x\leq2,
	\]
	with $u_1(1)=u_1(2)=0$.
	
	Figure~\ref{fig:counterexample_two_period} illustrates the resulting
	equilibrium dispatch and prices.
\begin{center}
	\begin{minipage}{\columnwidth}
		\centering
		\resizebox{\columnwidth}{!}{%
			\begin{tikzpicture}[
				>=Latex,
				bus/.style={
					circle,
					minimum size=12mm,
					line width=0.8pt,
					align=center,
					font=\small
				},
				storage/.style={
					rounded corners=1.5pt,
					minimum width=14mm,
					minimum height=8mm,
					draw=purple!70!black,
					fill=purple!12,
					line width=0.8pt,
					align=center,
					font=\scriptsize
				},
				load/.style={
					rounded corners=1.5pt,
					minimum width=13mm,
					minimum height=7mm,
					draw=gray!75!black,
					fill=gray!10,
					line width=0.8pt,
					align=center,
					font=\scriptsize
				}
				]
				
				\filldraw[
				rounded corners=3pt,
				fill=blue!3,
				draw=blue!45!black,
				line width=0.8pt
				]
				(0,0) rectangle (5.05,5.20);
				\node[font=\bfseries\small] at (2.525,4.8)
				{Period 1: charging};
				
				\node[
				bus,
				draw=blue!65!black,
				fill=blue!12
				] (b11) at (1.00,2.55)
				{Bus 1\\$\lambda_1(1)=4$};
				\node[
				bus,
				draw=orange!75!black,
				fill=orange!15
				] (b21) at (3.85,2.55)
				{Bus 2\\$\lambda_2(1)=13$};
				
				\node[font=\small, align=center] at (0.75,3.7)
				{$g_1(1)=4$};
				\node[font=\small, align=center] at (4,3.7)
				{$g_2(1)=1$};
				
				\draw[
				->,
				gray!80!black,
				line width=1.2pt
				]
				(b11.east) -- node[
				below,
				yshift=-2mm,
				font=\scriptsize,
				text=gray!80!black,
				align=center
				] {$f(1)=\overline f$}
				(b21.west);
				
				\draw[
				->,
				dashed,
				green!55!black,
				line width=1.2pt,
				bend left=55
				]
				(b11.north east) to node[
				above,
				font=\scriptsize,
				text=green!45!black,
				align=center
				] {FTR $1\!\to\!2$\\[-0.5mm]$\kappa_{(1,2)}(1)=4$}
				(b21.north west);
				
				\node[storage] (s1) at (2.55,0.63)
				{Storage\\$u_2(1)=+2$};
				\node[load] (d1) at (4.25,0.63)
				{Load\\$d_2(1)=3$};
				\draw[->, purple!70!black, line width=0.9pt]
				(b21.south west) -- (s1.north);
				\draw[->, gray!75!black, line width=0.9pt]
				(b21.south east) -- (d1.north);
				
				\filldraw[
				rounded corners=3pt,
				fill=green!3,
				draw=green!40!black,
				line width=0.8pt
				]
				(5.35,0) rectangle (10.40,5.20);
				\node[font=\bfseries\small] at (7.875,4.8)
				{Period 2: discharging};
				
				\node[
				bus,
				draw=blue!65!black,
				fill=blue!12
				] (b12) at (6.35,2.55)
				{Bus 1\\$\lambda_1(2)=2$};
				\node[
				bus,
				draw=orange!75!black,
				fill=orange!15
				] (b22) at (9.20,2.55)
				{Bus 2\\$\lambda_2(2)=2$};
				
				\node[font=\small, align=center] at (6.4,3.7)
				{$g_1(2)=1$};
				\node[font=\small, align=center] at (9.2,3.7)
				{$g_2(2)=0$};
				
				\draw[
				->,
				gray!80!black,
				line width=1.2pt
				]
				(b12.east) -- node[
				below,
				yshift=-2mm,
				font=\scriptsize,
				text=gray!80!black,
				align=center
				] {$f(2)=1$}
				(b22.west);
				
				\node[storage] (s2) at (7.9,0.63)
				{Storage\\$u_2(2)=-2$};
				\node[load] (d2) at (9.60,0.63)
				{Load\\$d_2(2)=3$};
				\draw[->, purple!70!black, line width=0.9pt]
				(s2.north) -- (b22.south west);
				\draw[->, gray!75!black, line width=0.9pt]
				(b22.south east) -- (d2.north);
			\end{tikzpicture}%
		}
		\captionsetup{hypcap=false}
		
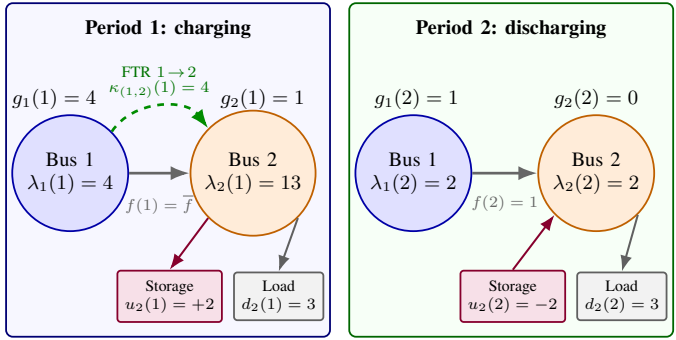
\captionof{figure}{Two-bus, two-period counterexample. At the equilibrium schedule,
			the ESA charges two units in period~1, binds the line, and creates the
			LMP spread that supports its FTR payoff. It discharges two units in
			period~2, when the line is uncongested.}
		\label{fig:counterexample_two_period}
	\end{minipage}
\end{center}

	\textcolor{black}{For a given storage schedule $x$, the period-1 net demand at bus~2 is
		$3+x$. Since the line capacity is $\overline f=4$, the line is uncongested
		for $0\leq x\leq1$ and binds for $1\leq x\leq2$. Therefore, the resulting
		ISO-clearing dispatch is
		\begin{equation}
			(g_1(1),g_2(1),f(1))
			=
			\begin{cases}
				(3+x,0,3+x), & 0\leq x\leq1,\\
				(4,x-1,4),  & 1\leq x\leq2.
			\end{cases}
			\label{eq:counterexample_dispatch_t1}
		\end{equation}
		In period~2, storage discharges $x$, so the net demand at bus~2 is $3-x$.
		This demand can be served entirely by the generator at bus~1 without congesting
		the line, and hence
		\begin{equation}
			(g_1(2),g_2(2),f(2))
			=
			(3-x,0,3-x).
			\label{eq:counterexample_dispatch_t2}
		\end{equation}
	}
	\textcolor{black}{Consequently, the LMPs are
		\begin{equation}
			(\lambda_1(1),\lambda_2(1))
			=
			\begin{cases}
				(3+x,3+x),&0\leq x\leq1,\\
				(4,4+9(x-1)),&1\leq x\leq2,
			\end{cases}
\nonumber
		\end{equation}
		and
		\begin{equation}
			(\lambda_1(2),\lambda_2(2))
			=
			(4-x,4-x).\nonumber
		\end{equation}
		Moreover, the corresponding total generation cost as a function of $x$ is obtained from
		\eqref{eq:counterexample_dispatch_t1} and
		\eqref{eq:counterexample_dispatch_t2}. For $0\leq x\leq1$, the dispatch is
		$(g_1(1),g_2(1))=(3+x,0)$ and $(g_1(2),g_2(2))=(3-x,0)$, so the total
		generation cost is $12-x+x^2$. For $1\leq x\leq2$, the dispatch is
		$(g_1(1),g_2(1))=(4,x-1)$ and $(g_1(2),g_2(2))=(3-x,0)$, so the total
		generation cost is $16-9x+5x^2$. Thus,
		\begin{equation}
			C^{\mathrm{gen}}(x)
			=
			\begin{cases}
				12-x+x^2, & 0\leq x\leq1,\\
				16-9x+5x^2, & 1\leq x\leq2.
			\end{cases}
			\label{eq:counterexample_generation_cost}
		\end{equation}
	}
	
	\textcolor{black}{Let the ESA hold a period-1 FTR obligation of quantity
		$\kappa_{(1,2)}(1)=4$ from bus 1 to bus 2 and no FTR in period 2. The ESA
		objective, including its energy-market payment and subtracting its FTR payoff,
		is
		\begin{equation}
			J(x)
			=
			\begin{cases}
				2x^2-x,&0\leq x\leq1,\\
				10x^2-45x+36,&1\leq x\leq2.
			\end{cases}
			\label{eq:counterexample_esa_objective}
		\end{equation}
		Indeed, for $x\geq1$, the energy-market payment is
		$x[\lambda_2(1)-\lambda_2(2)]=10x^2-9x$, whereas the FTR payoff is
		$4[\lambda_2(1)-\lambda_1(1)]=36(x-1)$. The minimum of the first branch of
		\eqref{eq:counterexample_esa_objective} is $-1/8$, attained at $x=1/4$.
		The derivative of the second branch is $20x-45<0$ on $[1,2]$, so that branch
		is minimized at $x=2$, where $J(2)=-14<-1/8$. Hence the unique equilibrium
		schedule is
		\[
		\mathbf u^\star:
		\qquad
		u_2^\star(1)=2,
		\qquad
		u_2^\star(2)=-2.
		\]
		At this schedule, the ESA pays $22$ in the energy market but receives $36$
		from its FTR position; thus, the FTR payoff strictly dominates the adverse
		energy-market payment.
	}
	\textcolor{black}{It remains to compare welfare. Since demand utility is fixed and normalized
		to zero, welfare equals minus total generation cost. Without storage
		($x=0$), total generation cost is
		\begin{equation}
			C^{\mathrm{gen}}(0)=12,
			\nonumber
		\end{equation}
		and hence $\mathrm{SW}^{\mathrm{cen}}(\mathbf0)=-12$. At the equilibrium
		$x=2$, total generation cost is
		\begin{equation}
			C^{\mathrm{gen}}(2)=18,
			\nonumber
		\end{equation}
		so $\mathrm{SW}(\mathbf u^\star)=-18$. Finally, under centralized storage dispatch, the ISO chooses $x\in[0,2]$ to
		minimize the total generation cost in
		\eqref{eq:counterexample_generation_cost}. The first branch is minimized at
		$x=1/2$, yielding total generation cost $47/4$. The unconstrained minimizer of
		the second branch is $x=9/10$, which lies outside the interval $[1,2]$; hence
		the minimum over the second branch is attained at $x=1$. Therefore, the
		centralized optimum is attained at $x=1/2$, with total generation cost
		$47/4$.}
	Therefore,
	\[
	\mathrm{SW}(\mathbf u^\star)
	=
	-18
	<
	-12
	=
	\mathrm{SW}^{\mathrm{cen}}(\mathbf0)
	\leq
	-\frac{47}{4}
	=
	\mathrm{SW}^{\mathrm{cen}}(\mathbf B),
	\]
	which establishes the existence claim in \eqref{eq:welfare_with_ftr}.
\end{IEEEproof}
\endgroup
\section{Proofs for Section~\ref{sec:iso_binding_pattern_control}}
\subsection{Proof of Proposition~\ref{prop:linear_safe_capacity_limits} and
	Corollary~\ref{cor:safe_participation_capacity_limits}}
\label{app:proof_capacity_limits}
\begin{IEEEproof}[Proof of Proposition~\ref{prop:linear_safe_capacity_limits}]
	Let $h_k$ denote the $k$-th entry of $\mathbf h_{\mathrm{des}}$. By Assumption~\ref{assump:polyhedral_desired_set}, the containment
	\(\mathcal U(\overline{\mathbf B})\subseteq\mathcal C_{\mathrm{des}}\)
	holds if and only if, for every \(k=1,\ldots,m_{\mathrm{des}}\),
\begin{equation}
	\max_{\mathbf u\in\mathcal U(\overline{\mathbf B})}
	\mathbf f_k^\top\operatorname{vec}(\mathbf u)
	\leq h_k.
	\label{eq:appendix_support_condition}
\end{equation}
	Using the bus-wise partition~\eqref{eq:desired_row_partition} and the
	separability of the storage constraints,
\begin{align*}
	\max_{\mathbf u\in\mathcal U(\overline{\mathbf B})}
	\mathbf f_k^\top\operatorname{vec}(\mathbf u)
	&=
	\sum_{i\in\mathcal N}
	\max_{\substack{
			\mathbf0\leq\mathbf L\mathbf u_i
			\leq\overline B_i\mathbf1\\
			\mathbf1^\top\mathbf u_i=0}}
	\mathbf f_{k,i}^\top\mathbf u_i.
\end{align*}
	If \(\overline B_i>0\), the substitution
	\(\mathbf u_i=\overline B_i\mathbf w_i\) gives
	\[
	\max_{\substack{
			\mathbf0\leq\mathbf L\mathbf u_i
			\leq\overline B_i\mathbf1\\
			\mathbf1^\top\mathbf u_i=0}}
	\mathbf f_{k,i}^\top\mathbf u_i
	=
	\overline B_i\xi_{k,i}.
	\]
	The same identity holds when \(\overline B_i=0\), because then
	\(\mathbf L\mathbf u_i=\mathbf0\), and the invertibility of \(\mathbf L\)
	implies \(\mathbf u_i=\mathbf0\).  Therefore
	\[
	\max_{\mathbf u\in\mathcal U(\overline{\mathbf B})}
	\mathbf f_k^\top\operatorname{vec}(\mathbf u)
	=
	\sum_{i\in\mathcal N}\xi_{k,i}\overline B_i.
	\]
	Substituting this equality into~\eqref{eq:appendix_support_condition} for
	every row \(k\) gives
$
	\mathcal U(\overline{\mathbf B})
	\subseteq\mathcal C_{\mathrm{des}}
	\quad\Longleftrightarrow\quad
	\bm\Xi\overline{\mathbf B}\leq\mathbf h_{\mathrm{des}},
$
	as claimed.
\end{IEEEproof}

\begin{IEEEproof}[Proof of Corollary~\ref{cor:safe_participation_capacity_limits}]
	Let
	\(\mathbf0\leq\mathbf B\leq\overline{\mathbf B}\) componentwise.  Every
	\(\mathbf u\in\mathcal U(\mathbf B)\) satisfies
	\[
	\mathbf0
	\leq\mathbf L\mathbf u_i
	\leq B_i\mathbf1
	\leq\overline B_i\mathbf1,
	\qquad
	\mathbf1^\top\mathbf u_i=0,
	\]
	for each \(i\), and hence
	\(\mathcal U(\mathbf B)\subseteq\mathcal U(\overline{\mathbf B})\).
	Feasibility of \(\overline{\mathbf B}\) for
	\eqref{eq:capacity_limit_design_lp} gives
	\(\bm\Xi\overline{\mathbf B}\leq\mathbf h_{\mathrm{des}}\).  Proposition
	\ref{prop:linear_safe_capacity_limits} then yields
	\(\mathcal U(\overline{\mathbf B})\subseteq\mathcal C_{\mathrm{des}}\).
	Combining the two inclusions proves~\eqref{eq:safe_capacity_guarantee}.
\end{IEEEproof}

\end{document}